\documentclass[10pt,twocolumn]{IEEEtran}
\usepackage{setspace}
\usepackage{cite}
\usepackage{graphicx}
\usepackage{subfigure}
\usepackage{amsmath}
\usepackage{amssymb}
\usepackage{amsthm}
\usepackage{amsmath}
\usepackage{algorithm,algorithmic}
\usepackage{booktabs}
\usepackage{multirow}
\usepackage{color}

\definecolor{darkblue}{rgb}{0,0,0.7}

\newcommand{\ip}[2]{\left\langle {#1}, {#2} \right\rangle}

\newcommand{\bdsb}[1]{\boldsymbol{#1}}

\newcommand{\dtau}{\Delta \tau}
\newcommand{\dw}{\Delta \omega}

\newtheorem{thm}{Theorem}
\newtheorem{lem}{Lemma}
\newtheorem{prop}{Proposition}

\newtheorem{definition}{Definition}
\newtheorem{cor}{Corollary}

\title{Compressive Link Acquisition in Multiuser Communications}
\author{
Xiao Li$^{\dagger}$, Andrea Rueetschi$^{\dagger}$, Anna Scaglione$^{\dagger}$ and Yonina C. Eldar$^{\ast}$
\thanks{$^{\dagger}$The authors are with the Department of Electrical and Computer Engineering, University of California, Davis 95616, USA, email : \{eceli,andrear,ascaglione\}@ucdavis.edu}
\thanks{$^{\ast}$Y. C. Eldar is with Department of Electrical Engineering, Technion, Israel Institute of Technology, Haifa 32000, Israel, email : yonina@ee.technion.ac.il}
}

\begin{document}
	
\graphicspath{{./figure/}{./sim_plots/}{./sim_draw}}

\maketitle
\begin{abstract}
An important receiver operation is to detect the presence specific preamble signals with unknown delays in the presence of scattering, Doppler effects and carrier offsets. This task,  referred to as ``link acquisition'', is typically a sequential search over the transmitted signal space. Recently, many authors have suggested applying sparse recovery algorithms in the context of similar estimation or detection problems. These works typically focus on the benefits of sparse recovery, but not generally on the cost brought by compressive sensing. Thus, our goal is to examine the trade-off in complexity and performance that is possible when using sparse recovery. To do so, we propose a sequential sparsity-aware compressive sampling (C-SA) acquisition scheme, where a compressive multi-channel sampling (CMS) front-end is followed by a sparsity regularized likelihood ratio test (SR-LRT) module.

The proposed C-SA acquisition scheme borrows insights from the models studied in the context of sub-Nyquist sampling, where a minimal amount of samples is captured to reconstruct signals with Finite Rate of Innovation (FRI). In particular, we propose an A/D conversion front-end that maximizes a well-known probability divergence measure, the average Kullback-Leibler distance, of all the hypotheses of the SR-LRT performed on the samples. We compare the proposed acquisition scheme vis-\`{a}-vis conventional alternatives with relatively low computational cost, such as the Matched Filter (MF), in terms of performance and complexity. Our experiments suggest that one can use the proposed C-SA acquisition scheme to scale down the implementation cost with greater flexibility than MF architectures. However, we find that they both have overall complexities that scale linearly with the search space despite of the compressed samples. Furthermore, it is shown that compressive measurements used in the SR-LRT at low SNR lead to a performance loss as one could expect given that they use less observations, while at high SNR on the other hand, the SR-LRT has better performance in spite of the compression.
\end{abstract}

\section{Introduction}
One of the critical receiver tasks in a multiuser scenario, referred to as {\it link acquisition}, is that of detecting the presence of signals, and identifying the {\it link parameters} (e.g., delays, carrier offsets) of an {\it unknown subset} $\mathcal{I}$ of active sources out of $I$ possible ones.  Similar to \cite{fletcher2009off,applebaum2011asynchronous}, we consider the case in which the set of active users $\mathcal{I}$ transmit known and distinct training preambles $\phi_i(t)$, $i\in\mathcal{I}$. Usually these preambles are designed to be fairly long so that their energy harvested at the receiver can rise above the noise. In this initial phase, the receiver is completely agnostic about the existing sources and tests the received signal $x(t)$ until it detects the presence of such signals, to establish the active user links. This needs to be done by accumulating observations and repeating the test sequentially. The acquired link information is essential for identifying the basic features of the received signal, so that the receiver can determine if it can decode the data after the training phase \cite{verdu1998multiuser,xie2010reduced} and refine the link parameter estimates using mid-ambles and decoded data. The term ``link acquisition" is equivalent to resolving the received signal space, which is characterized by the propagation delays and the carrier offsets. {Because we rely on sparsity of active sources in our problem, the scenario that is best suited for this is a downlink setting.}

\subsection{Related Works on Link Acquisition of Multiuser Signals}
We can categorize the algorithms that are used for link acquisition into two main groups. The first category acquires a {\it sufficient statistic} by directly sampling $x(t)$ at (or above) the Nyquist rate. The likelihood function associated with the sufficient statistic is then exploited to detect the presence of signals and determine the link parameters in the model given the set of active users $\mathcal{I}$. We refer to such techniques as Direct Sampling (DS) methods (e.g. \cite{sirkeci2004signal,buzzi2003code, qiu2001fast, wang2005decorrelating}).

A second approach, referred to as the Matched Filtering (MF) \cite{fishler2006spatial,tian2005glrt,xie2010reduced}, facilitates the search of both the {\it active set} $\mathcal{I}$ and link parameters by comparing the filtered outputs of the signal $x(t)$ from a bank of filters constructed by shifting and modulating the preamble $\phi_i(t-\tau)e^{-\mathrm{i}\omega t}$, each matching a sufficiently wide collection of points in the full parameter set $\mathcal{T}\times\mathcal{F}$ where $\tau\in\mathcal{T}$ and $\omega\in\mathcal{F}$ are the delay and Doppler spread respectively. MF is a prevalent choice in hardware implementations because of its simplicity. The MF approach can be implemented in the digital domain, where samples are projected onto the sampled version of  $\phi_i(t-\tau)e^{-\mathrm{i}\omega t}$, or in the analog domain, where the receiver performs filtering operations onto the templates $\phi_i(t-\tau)e^{-\mathrm{i}\omega t}$ directly in hardware. Specific details on these architectures are provided in Section \ref{sampling_acquisition}.

Classical algorithms take little advantage of the low dimensionality of the received signal space in storing and processing the observations to improve the performance or reduce complexity. Recently, there have been advances in exploiting {\it sparsity}, or the low dimensionality of the signal space, to improve receiver performance. One class of papers suggests using sparse signal recovery for the purpose of either detection or estimation. For instance, assuming that the signal is present, the results in \cite{fletcher2009off,applebaum2011asynchronous,xie2010reduced,zhu2011exploiting} deal with identification of the active users and/or estimation of signal parameters by creating a dictionary from the known templates $\phi_i(t)$ and viewing the signal $x(t)$ as a sparse linear combination of these element templates inside the dictionary. Without knowledge of signal presence within a specific observation window, the proposed detection schemes in \cite{duarte2006sparse,davenport2010signal,haupt2007compressive,wang2008subspace,paredes2009compressive} use generic compressed measurements to detect the presence of certain signals, starting from an abstract discrete model.  We refer to this general class of methods as the Sparsity-Aware (SA) approach. In these papers, delays and carrier offsets are not explicitly considered and the discrete observations are treated independently as a single snapshot from certain linear models, upon which SA algorithms are applied.

\subsection{Multiuser Signals with Finite Rate of Innovation (FRI)}
What is often neglected in existing DS and SA approaches is the acquisition of informative low rate discrete samples from the analog domain. As we mentioned, preamble sequences are usually fairly long and the receiver needs to sample the signal $x(t)$ at a fast rate and store them prior to processing. This can become a bottleneck in designing preambles so that they have the appropriate energy.

Reducing the sampling rate and the associated storage incurred at the A/D front-end is mostly the concern of another broad class of papers \cite{gedalyahu2010time,bajwa2011identification,vetterli2002sampling,maravic2005sampling,gedalyahu2011multichannel} on signals with a Finite Rate of Innovation (FRI) \cite{uriguen2011sampling}. In general, an FRI model has a {\it sparse} parametric representation. { Given the preamble $\phi_i(t)$ for each active user $i\in\mathcal{I}$ traveling through $R$ multipath channels, the class of signals $x(t)$ lies in a subspace with no more than $|{\mathcal I}|R$ dimensions, where each dimension has three unknowns (e.g., delay, carrier offset, channel coefficient), irrespective of its bandwidth and duration.}

The premier objective of FRI sampling is A/D conversion at sub-Nyquist rates for deterministic signal reconstructions. This objective is fundamentally different from what is of practical interest in link acquisition, which is to perform statistical inference. In this paper, we wish to harness similar cost reductions as in the FRI literature, while mitigating the detection performance losses that arise in the presence of noise due to the reduced number of observations. To this aim, we formulate the link acquisition problem as a Sparsity Regularized Likelihood Ratio Test (SR-LRT) that uses samples from a compressive multi-channel sampling (CMS) architecture. To enhance the acquisition performance, we optimize the sampling front-end by choosing sampling kernels that maximize a probability divergence measure of all the hypotheses in the test. We refer to the proposed link acquisition scheme as the Sparsity-Aware Compressive Sampling (C-SA) acquisition scheme. More specifically, we discuss in this paper
\begin{enumerate}
	\item a unified low-rate A/D conversion front-end using the proposed CMS architecture;
	\item a SR-LRT that uses compressive samples from the CMS architecture for sequential joint detection and estimation;
	\item the optimization of the CMS architecture for maximum average Kullback-Leibler (KL) distance of the SR-LRT;  	
	\item the comparison of the proposed C-SA scheme with the MF approach in terms of performance and costs.
\end{enumerate}	
This bridges the results pertaining to sparsity-aware estimation/detection \cite{applebaum2011asynchronous,fletcher2009off,xie2010reduced,duarte2006sparse,davenport2010signal,haupt2007compressive,wang2008subspace,paredes2009compressive}, the literature on analog compressed sensing and sub-Nyquist sampling \cite{tian2005glrt,gedalyahu2010time,wang2005decorrelating,bajwa2011identification, eldar2009compressed,xie2010reduced} and FRI sampling \cite{vetterli2002sampling,maravic2005sampling,uriguen2011sampling} such that sampling and acquisition operations are considered jointly.

To measure the benefits of the proposed C-SA scheme over other schemes, we analyze the practical trade-off between the implementation costs in physically acquiring samples and those invested computationally in sparse recovery. This is important to clarify the potential benefits of sub-Nyquist architectures in communication receivers in the link acquisition phase. These schemes often benefit from the denoising capabilities of SA algorithms (as well documented in \cite{applebaum2011asynchronous,fletcher2009off,xie2010reduced,duarte2006sparse,davenport2010signal,haupt2007compressive,wang2008subspace,paredes2009compressive}) but must loose sensitivity due to the fact that they do not use sufficient statistics for the receiver inference.

The question we consider is, therefore, {\it what is there to gain: implementation costs or performance?} Our numerical experiments indicate that the main advantage of the proposed scheme is that it enables the designer to find an adequate operating point for link acquisition such that processing requirements and complexity of the receiver can be reduced to an acceptable level without significantly sacrificing acquisition performance compared with the MF architecture. We also confirm numerically that the compressive samplers proposed in the CMS architecture harvest highly informative samples for the SR-LRT in terms of estimation and detection performance.

{ \subsection{Notation and Paper Organization}
We denote vectors and matrices by boldface lower-case and boldface upper-case symbols and the set of real (complex) numbers by $\mathbb{R}$ ($\mathbb{C}$). We denote sets by calligraphic symbols, where the intersection and the union of two sets $\mathcal{A}$ and $\mathcal{B}$ are written as $\mathcal{A}\bigcap\mathcal{B}$ and $\mathcal{A}\bigcup\mathcal{B}$ respectively. The operator $|\mathcal{A}|$ on a discrete (continuous) set takes the cardinality (measure) of the set. The magnitude of a complex number $x$ is denoted by $|x|=\sqrt{xx^\ast}$, where $x^\ast$ is the conjugate of the complex number $x$. The transpose, conjugate transpose, and inverse of a matrix $\mathbf{X}$ are denoted by $\mathbf{X}^T$, $\mathbf{X}^H$ and $\mathbf{X}^{-1}$, respectively. The inner products between two vectors $\mathbf{x},\mathbf{y}\in\mathbb{C}^{N\times 1}$ and between two continuous functions $f(t),g(t)$ in $L_2(\mathbb{C})$ are defined accordingly as $\ip{\mathbf{x}}{\mathbf{y}} = \sum_{n=1}^N y_n^\ast x_n$ and $\ip{f(t)}{g(t)}=\int_{-\infty}^{\infty}g^\ast(t)f(t)\mathrm{d}t$. The $\mathbf{W}$-weighted $\ell_2$-norm of a vector $\mathbf{x}$ {with a positive definite matrix $\mathbf{W}$} is denoted by $\left\|\mathbf{x}\right\|_{\mathbf{W}}=\sqrt{\mathbf{x}^H\mathbf{W}\mathbf{x}}$, and the conventional $\ell_2$-norm is written as $\left\|\mathbf{x}\right\|$. The $L_2$-norm of a continuous-time signal $f(t)\in L_2(\mathbb{C})$ is computed as $\left\|f(t)\right\|=\sqrt{\ip{f(t)}{f(t)}}$. 
}

The paper is organized as follows. Section \ref{model_based_link_acquisition} introduces our received signal model. We discuss related works on link acquisition in Section \ref{sampling_acquisition}. The CMS architecture for compressive acquisition is considered in Section \ref{compressive_acquisition}. Using the compressive samples obtained from the CMS architecture, we develop the SR-LRT algorithm for C-SA link acquisition scheme in Section \ref{SR_GLRT}. We then optimize the compressive samplers in the CMS in Section \ref{optimization_B}. Simulations demonstrating the performance are presented in Section \ref{numerical_results}. The overall cost of the proposed C-SA scheme is compared against conventional MF schemes in terms of storage and computational costs in Section \ref{complexity}.

\section{Signal Model for Link Acquisition}\label{model_based_link_acquisition}
In every communication standard, a key control sequence in the training phase is the initial preamble. The receiver models the corresponding observation by assuming that each $i\in\mathcal{I}$ from the unknown active set transmits a specific preamble $\phi_i(t)$. This transmission is followed by the mid-ambles and data frames. A common choice for such a preamble in multiuser communications is a linearly pulse modulated sequence with a chip rate $1/T$ { close to the signal bandwidth and equal to the minimum Nyquist-rate}
\begin{align}\label{preamble}
	\phi_i(t) = \sum_{m=0}^{M-1}a_i[m]g\left(t-mT\right).
\end{align}
Here $g(t)$ is the pulse shaping filter (chip) and $\{a_i[m]\}_{m=1}^M$ is typically a long preamble sequence $M\gg 1$ for each user.

\subsection{Received Signal Model}
Then the observation at the receiver can be written as
\begin{align}\label{eq.orig}
	x(t) = \sum_{i\in\mathcal{I}} \sum_{r=1}^{R} h_{i,r}\phi_i(t-t_{i,r})e^{\mathrm{i}\omega_{i,r}t} + v(t),
\end{align}
where $t_{i,r}$ is the unknown propagation delay of the $i$th user in the $r$th multipath, $|\omega_{i,r}|\leq \omega_{\rm max}$ is the Doppler frequency upper bounded by the maximum Doppler spread $\omega_{\rm max}$, and $h_{i,r}$ is the unknown channel fade. Without loss of generality, we assume that the maximum multipath order $R$ is known and the noise component $v(t)$ is a white Gaussian process with $\mathbb{E}\{v(t)v^\ast(s)\}=\sigma^2\delta(t-s)$. Our problem is to detect the presence of the active user set $\mathcal{I}$ and the corresponding link parameters $\{h_{i,r},t_{i,r},\omega_{i,r}\}$ for $i\in\mathcal{I}$ and $r=1,\cdots,R$.


Since the propagation delays $t_{i,r}$ are unknown and possibly large, the typical A/D front-end for link acquisition is sequential. The acquisition scheme produces test statistics every $D$ units of time, where $D$ is the shift in the time reference for detections. At every shift $t=nD$, the receiver decides whether the signal $x(t)$ is present at or after $t=nD$. 
For convenience, we denote $t_0=\min \{t_{i,r}\}_{i\in\mathcal{I}}^{r=1,\cdots,R}$ as the delay of the first arrival path among all users. Let
\begin{align}\label{ell}
    \ell = \left\lfloor t_0/D \right\rfloor
\end{align}
be the shift that matches best with signal arrival and
\begin{align}\label{delay_decomp}
    \tau_{i,r} = t_{i,r} - \ell D \geq 0
\end{align}
be the {\it composite delay}, where $0\leq\tau_{i,r}\leq \tau_{\rm max}$ and $\tau_{\max}$ is the composite delay spread. {Note that $\tau_{i,r}=(t_0-\ell D)+(t_{i,r}-t_0)$, where the first term is the fractional delay within $[0,D)$ while the second term is a multipath delay relative to the first arrival path, which is bounded by the channel delay spread $\tilde \tau_{\max}$. This implies that  $\tau_{i,r}\leq D+\tilde \tau_{\max}$.
Given the multipath delay spread $\tilde{\tau}_{\max}$ and the shift size $D$, we can obtain the composite delay spread $\tau_{\max}=D+\tilde \tau_{\max}$ as the search space to fully capture the signal at least in the $\ell$th shift.} 

This allows us to express \eqref{eq.orig} equivalently as
\begin{align}\label{orig_sig}
	x(t) &= \sum_{i\in\mathcal{I}} \sum_{r=1}^{R} h_{i,r} \phi_i(t-\ell D-\tau_{i,r}) e^{\mathrm{i}\omega_{i,r}t} + v(t).
\end{align}
After these considerations, it is clear that the search spaces of delays and Doppler frequencies for each shift $n$ are respectively $\mathcal{T}\triangleq[0,\tau_{\max}]$ and $\mathcal{F}\triangleq[-\omega_{\max},\omega_{\max}]$.

\subsection{Goal of Link Acquisition}
Link acquisition is typically formulated as composite hypothesis tests with unknown link parameters, where the likelihood ratio between the signal hypothesis and the noise hypothesis is the test statistic for the detection task. Note that there could be multiple values of $n\neq \ell$ that lead to valid positive detections, where for a given $\ell$, the relative composite delay with respect to the $n$th shift would be
\begin{align}\label{modified_delay}
    \tau_{i,r}^{(n)}=\tau_{i,r}+(\ell-n)D.
\end{align}
{Therefore, when the signal is captured in an earlier shift $n<\ell$, the relative composite delay would be greater than $\tau_{i,r}$, and if it is captured in a later shift $n>\ell$ the relative composite delay would be smaller than $\tau_{i,r}$.} In order to single out the best reference shift, the receiver will have to compare a sequence of $N_0$ test statistics after the first positive detection at $n=N_\eta$, and choose the particular shift $\ell_\star$ that maximizes the likelihood ratio. We call $\ell_\star$ the \emph{maximum likelihood ratio} (MLR) shift. {The look-ahead horizon $N_0$ can be chosen considering the type of sampling kernels, the preambles $\phi_i(t)$'s, and the delay spread $\tau_{\max}$, making reasonable approximations about the durations of the signals\footnote{To give a rule-of-thumb, suppose that the receiver streams the observation by $D$ units of time in every shift and in each shift, a portion of the analog signal $x(t)$ with length $T_{\rm observe}$ is observed by the receiver. Note that $T_{\rm observe}>D$ is required such that no portions of the signal $x(t)$ is missed in any shift. Let the temporal support of the preambles be finitely limited by $T_{\rm preamble}$ and it satisfies $T_{\rm observe}\geq T_{\rm preamble}+\tau_{\max}$, then the look ahead horizon can be chosen as $N_0=\lceil T_{\rm observe}/D\rceil.$ }}.

\begin{definition}
{\bf Link acquisition} refers to
\begin{enumerate}
	\item locating the MLR shift $\ell_\star$;
	\item identifying the set of active users $\mathcal{I}$ in the $\ell_\star$th shift;
	\item resolving the delay-Doppler pairs $\{\widehat{\tau}_{i,r},\widehat{\omega}_{i,r}\}$ for $i\in\mathcal{I}$ and $r=1,\cdots,R$.
\end{enumerate}
\end{definition}

{ Usually, the preamble signals $\phi_i(t)$'s have large energy, so that they can rise above the receiver noise. Given that the average power is constant, the $\phi_i(t)$ typically last much longer (i.e., $M$ is large) than subsequent mid-ambles or spreading codes that modulate data. For a typical wireless application such as GPS or IS-95/IMT-2000, transmitters continuously send out preamble sequences with length on the order of $M = 20 \times 1023$ \cite{qaisar2011cross} or $M = 32768$ \cite{kim2001robust}, respectively. This means that in order to detect the presence of such preambles and acquire the synchronization parameters, architectures using DS, MF or D-SA approaches would have to store a large amount of data to process in a sequential manner.} This phase is crucial to properly initialize any channel tracking that ensues. In Section \ref{sampling_acquisition}, we provide details on the A/D architectures and the corresponding post-processing for conventional link acquisition schemes. We then present the proposed CMS architecture and the C-SA acquisition scheme in Section \ref{compressive_acquisition}.

\section{Existing Architectures for Link Acquisition}\label{sampling_acquisition}

For future use, we let the Nyquist rate of the signal $x(t)$ be $f_{\textrm{\tiny NYQ}} = 2\mathcal{W} +\omega_{\max}/\pi$ with $\mathcal{W}$ being the maximum single-sided bandwidth of $\phi_i(t)$, $i=1,\cdots,I$. 

\subsection{Direct Sampling (DS)}
In DS schemes, the received analog signal $x(t)$ is sampled by projecting it onto an ideal series of Dirac's deltas, every $T_s = 1/f_s \leq 1/f_{\textrm{\tiny NYQ}}$, i.e.
\begin{align}
	c_{\textrm{\tiny DS}}[w] = \ip{x(t)}{\delta(t-wT_s)} = x\left(wT_s\right).
\end{align}
At the $n$th shift, DS schemes use the most recent ${{W}}$ Nyquist samples for every shift $D=NT_s$
\begin{align}\label{nyq_freq}
	\mathbf{c}_{\textrm{\tiny DS}}[n] = \left[c_{\textrm{\tiny DS}}[nN],\cdots,c_{\textrm{\tiny DS}}[nN+({{W}}-1)]\right]^T
\end{align}
to perform the detection. {This is a sliding window operation of $W$ samples where in each shift, the most obsolete $N$ samples are replaced with the latest $N$ samples, where the number of sample per shift satisfies $W\geq N$ such that no samples are missed between shifts}. Based on \eqref{eq.orig}, the samples $\mathbf{c}_{\textrm{\tiny DS}}[n]$ can be expressed as
\begin{align}\label{oversampling}
	\mathbf{c}_{\textrm{\tiny DS}}[n] &= \bdsb{\Phi}_{\mathcal{J}}\left(\bdsb{\tau}_{\mathcal{J}},\bdsb{\omega}_{\mathcal{J}}\right)\mathbf{h}_{\mathcal{J}} + \mathbf{v}[n],
\end{align}
where the parameters
\begin{align}
	\bdsb{\tau}_{\mathcal{J}} &\triangleq[\cdots,\tau_{j,1}^{(n)},\cdots,\tau_{j,R}^{(n)},\cdots]^T\\
	\bdsb{\omega}_{\mathcal{J}} &\triangleq[\cdots,\omega_{j,1},\cdots,\omega_{j,R},\cdots]^T\\
	\mathbf{h}_{\mathcal{J}} &\triangleq[\cdots, h_{j,1},\cdots,h_{j,R},\cdots]^T
\end{align}
represent the $|\mathcal{J}|R$ residual delays, Doppler and channel coefficients corresponding to the set of users $\mathcal{J}\subseteq\{1,\cdots,I\}$ in the $n$th shift, and the set of active users $\mathcal{J}$ {varies with $n$ depending on which user component is captured in that shift. We omit the argument $n$ to make the notations lighter.} The vector $\mathbf{v}[n]$ contains the noise samples $\left[\mathbf{v}[n]\right]_w=v(nD + wT_s)$, and $\bdsb{\Phi}_{\mathcal{J}}\left(\bdsb{\tau}_{\mathcal{J}},\bdsb{\omega}_{\mathcal{J}}\right)$ is a ${{W}}\times |\mathcal{J}|R$ sub-matrix of the complete ${{W}}\times IR$ matrix $\bdsb{\Phi}\left(\bdsb{\tau},\bdsb{\omega}\right)$, from which we extract columns $j\in\mathcal{J}$. The full matrix $\bdsb{\Phi}\left(\bdsb{\tau},\bdsb{\omega}\right)$ is defined by
\begin{align}
	 \Big[\bdsb{\Phi}\left(\bdsb{\tau},\bdsb{\omega}\right)\Big]_{w,(i-1)R+r}
	\triangleq
	\phi_i\left(wT_s - L_gT - \tau_{i,r}\right)e^{\mathrm{i}\omega_{i,r} wT_s},\nonumber
\end{align}
where {$L_gT$ is the duration of $g(t)$ on each of both sides}\footnote{In general, a pulse has finite durations only if $g(t)$ is not bandlimited. If the pulse is bandlimited, the pulse $g(t)$ consists of side-lobes of length $T$ and is usually truncated. The parameter $L_g$ in this case will be chosen so that a the model has modest approximation error.}. 

Using the $n$th shift of observations, link acquisition amounts to performing the following composite hypothesis test
\begin{align*}
	\mathcal{H}_{\mathcal{J}} : \mathbf{c}_{\textrm{\tiny DS}}[n] &= \bdsb{\Phi}_{\mathcal{J}}\left(\bdsb{\tau}_{\mathcal{J}},\bdsb{\omega}_{\mathcal{J}}\right)\mathbf{h}_{\mathcal{J}} +\mathbf{v}[n],\\
    \mathcal{H}_{\varnothing} : \mathbf{c}_{\textrm{\tiny DS}}[n] &= \mathbf{v}[n],
\end{align*}
with unknown parameters $\mathcal{J}$, $\bdsb{\tau}_{\mathcal{J}}$, $\bdsb{\omega}_{\mathcal{J}}$ and $\mathbf{h}_{\mathcal{J}}$. 
The Generalized Likelihood Ratio Test (GLRT) is then typically used. This test requires solving the following non-linear least squares estimation (NLLSE) problem
\begin{align*}
	 \left\{\widehat{\mathcal{I}},\bdsb{\widehat{\tau}}_{\widehat{\mathcal{I}}},\bdsb{\widehat{\omega}}_{\widehat{\mathcal{I}}}, \mathbf{\widehat{h}}_{\widehat{\mathcal{I}}}\right\} =
	 \underset{\mathcal{J},\bdsb{\tau}_{\mathcal{J}},\bdsb{\omega}_{\mathcal{J}},\mathbf{h}_{\mathcal{J}}}{\arg\min} ~ \left\|\mathbf{c}_{\textrm{\tiny DS}}[n] - \bdsb{\Phi}_{\mathcal{J}}\left(\bdsb{\tau}_{\mathcal{J}},\bdsb{\omega}_{\mathcal{J}}\right)\mathbf{h}_{\mathcal{J}}\right\|,
\end{align*}
over all possible $\mathcal{J}$, $(\bdsb{\tau}_{\mathcal{J}},\bdsb{\omega}_{\mathcal{J}})\in\mathcal{T}^{|\mathcal{J}|}\times\mathcal{F}^{|\mathcal{J}|}$, $\mathbf{h}_{\mathcal{J}}\in\mathbb{C}^{|\mathcal{J}|}$ to compute the generalized likelihood ratio
\begin{align}\label{likelihood_ratio_GLRT}
	\eta_{\textrm{\tiny DS}}(n)
	=
	 \frac{\mathbb{P}\left(\mathcal{H}_{\widehat{\mathcal{I}}}\right)}{\mathbb{P}\left(\mathcal{H}_{\varnothing}\right)}
\end{align}
with estimates $\left\{\widehat{\mathcal{I}},\bdsb{\widehat{\tau}}_{\widehat{\mathcal{I}}},\bdsb{\widehat{\omega}}_{\widehat{\mathcal{I}}}, \mathbf{\widehat{h}}_{\widehat{\mathcal{I}}}\right\}$ obtained at every shift $t=nD$. The expression of the generalized likelihood ratio is given in \cite{kayfundamentals} for cases when the noise variance $\sigma^2$ of $\mathbf{v}[n]$ is known and unknown. Using the corresponding ratio as test statistics, the receiver checks if the test statistic satisfies $\eta_{\textrm{\tiny DS}}(n)\geq \eta_0$ for some properly chosen threshold $\eta_0$. Without loss of generality, we consider the most general case where $\sigma^2$ is unknown. In this case, the generalized likelihood ratio is obtained as
\begin{align}
	\eta_{\textrm{\tiny DS}}(n)
	=
	\frac{\left\|\mathbf{c}_{\textrm{\tiny DS}}[n]\right\|^{2{{W}}}}{\left\|\mathbf{c}_{\textrm{\tiny DS}}[n] - \bdsb{\Phi}_{\widehat{\mathcal{I}}}\left(\bdsb{\widehat{\tau}}_{\widehat{\mathcal{I}}},\bdsb{\widehat{\omega}}_{\widehat{\mathcal{I}}}\right) \mathbf{\widehat{h}}_{\widehat{\mathcal{I}}}\right\|^{2{{W}}}}.
\end{align}
Denote the first shift that passes the GLRT as 
\begin{align}
	N_{\eta}\triangleq\min \left\{\underset{n}{\arg}~\eta_{\textrm{\tiny DS}}(n) \geq \eta_0\right\},
\end{align}
then the MLR shift $\ell_\star$ is given by
\begin{align}\label{GLRT_best_window}
	\ell_\star &= \arg\underset{n}{\max}
	~ \eta_{\textrm{\tiny DS}}(n),\quad n=N_{\eta},\cdots,N_{\eta}+N_0.
\end{align}

Obviously, the test described above is intractable in general, since there are $2^I$ hypotheses at each shift $t=nD$ to explore, and for each of them, there is an NLLSE problem to solve. Therefore in practice, DS acquisition schemes either deal with the known user case $\mathcal{J}=\mathcal{I}$ or assume the full set $\mathcal{J}=\{1,\cdots,I\}$ during detection, followed by NLLSE for that specific user set. When the set of active users $\mathcal{I}$ is unknown, alternatives are {\it Matched Filtering} (MF) and {\it Sparsity-Aware} (SA) approaches. The C-SA scheme in this paper is an instance of the SA technique, which performs sequential detection and estimation via the proposed SR-LRT using sub-Nyquist samples from the proposed CMS architecture. We next describe the MF approach and then the SA approach.

\subsection{Matched Filtering (MF)}\label{MF_approach}
The MF receiver is a widely used architecture in practice because of its ease in finding the link parameters by observing the outputs of the MF filterbank constructed from the {\it MF templates} $\phi_i(t) e^{\mathrm{i}k\Delta\omega t}$ for some $i$ and $k$. Since the size of the filterbank has to be finite, therefore it is usually assumed that $\tau_{i,r}\approx q_{i,r}\Delta \tau$  and  $\omega_{i,r}\approx k_{i,r}\Delta\omega$ for some integers $q_{i,r}$ and $k_{i,r}$ with a certain resolution $\Delta\tau=\tau_{\rm max}/Q$ and $\Delta\omega=\omega_{\rm max}/K$. The search spaces for the MF receiver then become $\mathcal{Q} = \{0,1,\cdots,Q-1\}$ and $\mathcal{K} = \{-K,\cdots,K\}$, which is the discrete counterpart of the continuous search space $\mathcal{T}\times\mathcal{F}$.

The MF receiver is a popular choice for multiuser acquisition \cite{verdu1998multiuser}, for example, in GPS receivers \cite{li2012gps} or CDMA receivers \cite{kim2001robust}. Its comparison with the C-SA acquisition scheme using the CMS architecture we propose in this paper is particularly insightful because, although the MF front-end requires a large filterbank, the post-processing of its outputs is very simple. One could certainly perform more complex post-processing to enhance its performance. For example, the Orthogonal Matching Pursuit (OMP) algorithm in our C-SA scheme can be applied on the MF outputs for this purpose. However, in that case, as illustrated in Section \ref{complexity}, the resulting scheme will have much higher storage cost and computational complexity requirements, which render the merits of the MF approach meaningless. More importantly, the OMP technique can be directly applied to the Nyquist samples, as done in SA methods, making the MF stage superfluous\footnote{Strictly speaking, the OMP technique performs a MF stage in its first iteration. The subsequent iterations can be viewed as applying successive interference cancellation (SIC) in multiuser communications.}.

The MF obtains the decision statistics by passing $x(t)$ through a bank of $P=I|\mathcal{K}|$ MF templates, and sampling the outputs every $\Delta\tau$. To be consistent with the sequential structure in \eqref{orig_sig} and the DS method, the MF shifts its templates every $D={{N}}T_s$, and samples the outputs every $\Delta\tau = T_s \leq 1/f_{\textrm{\tiny NYQ}}$. The MF output corresponding to the $i$th user at the $k$th discrete frequency $\omega = k\Delta\omega$ is obtained as
\begin{align}
	c_{i,k}[w] = \ip{x(t)}{\phi_i(t-wT_s)e^{\mathrm{i}k\Delta\omega (t-wT_s)}}.
\end{align}
Oftentimes, the filtering process is implemented in the digital domain using the samples $\mathbf{c}_{\textrm{\tiny DS}}[n]$ in \eqref{nyq_freq}. For consistency, we proceed with the description in the analog domain. At the $n$th shift, the samples used for detections can be stacked into an $I|\mathcal{K}|\times |\mathcal{Q}|$ exhaustive MF output array
\begin{align}\label{sample_array}
	\mathbf{C}_{\textrm{\tiny MF}}[n]
	=
	\begin{bmatrix}
		\vdots & \cdots & \vdots \\
		c_{1,k}[n{{N}}] & \cdots & c_{1,k}[n{{N}}+Q-1]\\
		\vdots & \cdots & \vdots \\
		c_{I,k}[n{{N}}] & \cdots & c_{I,k}[n{{N}}+Q-1]\\
		\vdots & \cdots & \vdots
	\end{bmatrix}.
\end{align}
Then, the MF receiver uses $\mathbf{C}_{\textrm{\tiny MF}}[n]$ as test statistics and performs the test on each user as follows
\begin{align*}
	\left| c_{i,k}\left[nN + q\right] \right|\geq \rho_i,~i=1,\cdots,I,~k\in\mathcal{K},~q\in\mathcal{Q},
\end{align*}
where $\rho_i$ is the chosen detection threshold for each user.

Denote the active user set at the $n$th shift as $$\widehat{\mathcal{I}}=\left\{i: \left|c_{i,k}\left[nN + q\right] \right|\geq \rho_i,~\forall i,k,q\right\}$$ and the shift that triggers the first positive detection as
\begin{align}
	N_{\eta}\triangleq \min\left\{\underset{n}{\arg}~\left|c_{i,k}\left[nN + q\right]\right|\geq\rho_i,\forall i,k,q\right\}.
\end{align}
The MLR shift is obtained by locating the maximum output
\begin{align*}
	\ell_\star \triangleq \arg \underset{n}{\max}\left\{\underset{i,k,q}{\max}~\left|c_{i,k}\left[nN + q\right]\right|\right\}, 
\end{align*}
where $n=N_{\eta},\cdots,N_{\eta}+N_0$.

Given the multipath order $R$, the delay-Doppler pairs at the $n$th shift are pinpointed by the $R$ strongest outputs for all detected users $i\in\widehat{\mathcal{I}}$ over the search space $k\in\mathcal{K}$ and $q\in\mathcal{Q}$. For convenience, we denote the strongest path by $\left(k_{i,1},q_{i,1}\right)$ for the $i$th user $i\in\widehat{\mathcal{I}}$ and rank the outputs by magnitudes
\begin{align}\label{MF}
	& |c_{i,k_{i,1}}\left[\ell_\star N+ q_{i,1}\right]|
	> |c_{i,k_{i,2}}\left[\ell_\star N + q_{i,2}\right]|
	> \cdots \\
	&\hspace{3cm} > |c_{i,k_{i,R}}\left[\ell_\star N + q_{i,R}\right]| > \cdots\nonumber
\end{align}
for each user $i\in\widehat{\mathcal{I}}$. The delay-Doppler pairs are identified as
\begin{align}
	\mathcal{M}_i\triangleq \left\{ (k_{i,r}, q_{i,r}) : r=1,\cdots,R\right\},
\end{align}
which give the following link parameters
\begin{align*}
	\hat{\tau}_{i,r} = q_{i,r}\Delta\tau, ~\hat{\omega}_{i,r} = k_{i,r} \Delta \omega,\quad (k_{i,r}, q_{i,r})\in \mathcal{M}_i.
\end{align*}

Although the MF approach shows an advantage in its post-processing and implementation, it has a few drawbacks:
\begin{itemize}
	\item[i)] the size of the MF filterbank scales with the number of users $I$ and the parameter set $|\mathcal{K}|$;
	\item[ii)] digital implementation requires high rate processing, increasing storage and pipelining\footnote{Pipelining refers to timely processing of the samples that stream into the system per unit of time.} cost;
	\item[iii)] the MF samples $c_{i,k}[\cdot]$ contain the interference from different users and multipath components;
\end{itemize}
During { the link acquisition phase}, the effect of interference (iii) is mitigated by using wideband pulses $g(t)$. During the data detection phase, multipath and multiuser interferences are dealt with using a RAKE type receiver and the interference is tackled either by using linear multiuser receivers or, in some cases, using successive interference cancelation (SIC) or even maximum likelihood multiuser detection \cite{verdu1998multiuser}. Typically, the complexities of these schemes for data detection grow rapidly with respect to the size of the MF filterbank (i) and the sampling rate (ii). Since data detection is conducted after the link acquisition, the uncertainties about the set of active users, their delays and Doppler frequencies have already been resolved and therefore, these tasks become more manageable.


\begin{figure*}
\centering
\includegraphics[width=0.7\linewidth]{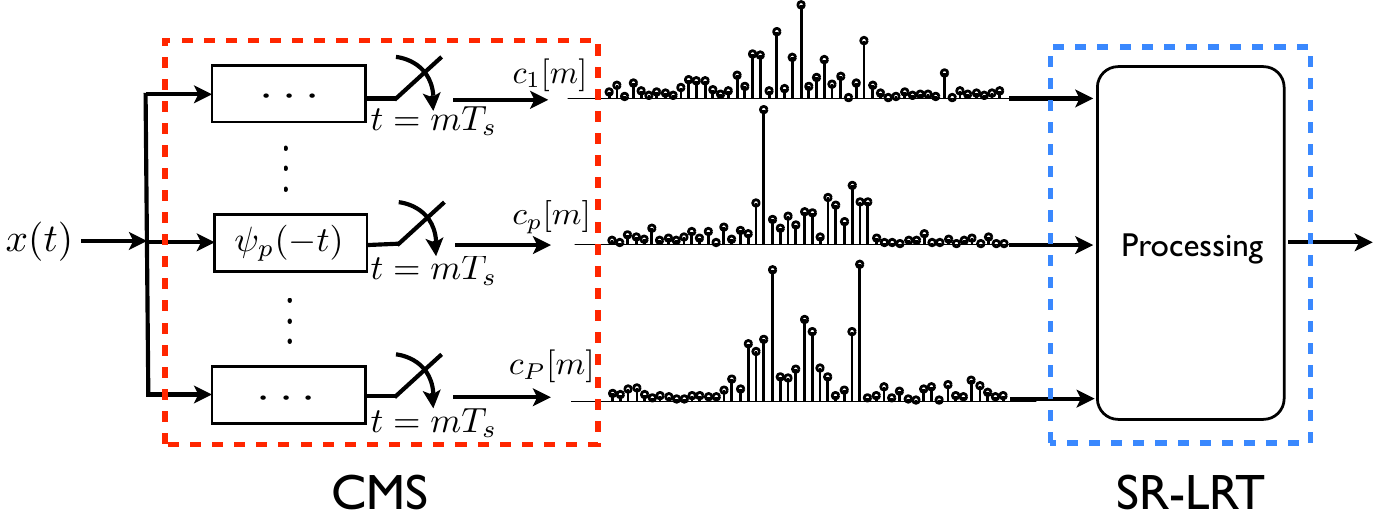}
\caption{The compressive samples obtained in the CMS architecture and the SR-LRT in the proposed C-SA acquisition scheme.}\label{CS}
\end{figure*}

\subsection{Sparsity Aware (SA) Approach}\label{sampling_acquisition_SA}
Instead of simply observing and ranking the MF outputs, many recent works have proposed the idea of {\it compressed sensing} or {\it sparse recovery} to solve estimation and detection problems. For the purpose of {\it user identification and parameter estimation}, one approach is to approximate \eqref{oversampling} by a sparse model with a dictionary constructed from the ensemble of possible templates $\phi_i(t)$ \cite{fletcher2009off,xie2010reduced} and/or discretized delays $\bdsb{\tau}$ \cite{applebaum2011asynchronous,daniele2010sparsity} (similar to the MF templates), where the joint recovery of active users and unknown parameters is relaxed as a {\it sparse estimation} problem. These sparse methods, which we call the sparsity-aware direct sampling (D-SA) scheme, usually require Nyquist rate samples and assume that the signal is already present (i.e., the MLR shift $\ell_\star$ is known). For clarity, D-SA scheme should not be confused with DS schemes, because DS schemes do not resort to sparsity-aware approaches based on discretization of analog parameters and require a non-linear search over the parameter space $\mathcal{T}\times\mathcal{F}$.

On the other hand, aiming at {\it signal presence detection} rather than identifying the active users and recovering the parameters, \cite{duarte2006sparse,davenport2010signal,haupt2007compressive,wang2008subspace,paredes2009compressive} reduce the number of samples required for the test by using a linear {\it compressor} on the block of given discrete observations. Since the compressor used in this class of detection schemes bears certain resemblances with the CMS architecture in the C-SA scheme we propose in this paper, we also categorize this method as the C-SA scheme just to avoid any confusion. Last but not least, in terms of the acquisition front-end, the D-SA scheme is a special case of C-SA scheme with a compressor that is an identity matrix.

A distinctive difference between the C-SA scheme in this paper and those in \cite{duarte2006sparse,davenport2010signal,haupt2007compressive,wang2008subspace,paredes2009compressive} is that the C-SA scheme proposed in this paper unifies the sequential signal detection, identification of active users and estimation of parameters by using the compressive samples obtained from a flexible multi-rate A/D architecture. On the other hand, the C-SA scheme in \cite{duarte2006sparse,davenport2010signal,haupt2007compressive,wang2008subspace,paredes2009compressive} directly starts from an abstract discrete model that is already sampled. Last but not least, the sampling kernels in the proposed CMS architecture are further optimized with respect to the estimation and detection performance in terms of the average KL distance of the hypotheses in the SR-LRT.


\section{Compressive Sequential Link Acquisition}\label{compressive_acquisition}
\subsection{Compressive Multichannel Sampling (CMS)}
We propose to use the A/D front-end in Fig. \ref{CS}, typical in FRI sampling \cite{gedalyahu2010time,eldar2009compressed,vetterli2002sampling,maravic2005sampling,uriguen2011sampling} for this work, which samples the signal every $t=nD$ by a $P$-channel filterbank
\begin{align}\label{gen_sampling}
	c_p[n] &\triangleq \ip{x(t)}{\psi_p(t-nD)},\quad p=1,\cdots,P.
\end{align}
We call this architecture the Compressive Multichannel Sampling (CMS) module, which forms the A/D conversion front-end of the proposed C-SA acquisition scheme.

Note that \eqref{gen_sampling} can also be implemented in the digital domain by performing linear projections of the discrete signal $\mathbf{c}_{\textrm{\tiny DS}}[n]$ in \eqref{oversampling}. This means that the CMS architecture becomes part of the post-processing of the Nyquist samples of $x(t)$, which lowers the storage and computation requirements as illustrated in Section \ref{complexity}. Similar derivations can be done in discrete time, but the advantage of using the analog description is that we do not necessarily have to target bandlimited signals. In the FRI literature \cite{uriguen2011sampling,vetterli2002sampling,maravic2005sampling,gedalyahu2010time,bajwa2011identification,eldar2009compressed}, in the absence of noise, the sampling rate required for the unique reconstruction of the signal in \eqref{orig_sig} is the { {\it number of degrees of freedom} of the signal $x(t)$ per shift $D$}, equal to the number of unknowns $\{\tau_{i,r},\omega_{i,r},h_{i,r}\}_{i\in\mathcal{I},r=1,\cdots,R}$. This amounts to $P_{\rm min}=3|\mathcal{I}|R$, which can be much less than what is needed in the MF approach $P_{\rm min} \ll P_{\rm MF} ={ I|\mathcal{K}|}$ when the number of active users is not large  $|\mathcal{I}|\ll I$, or the number of multipaths $R$ is much less than the dimension of the search space for Doppler $|\mathcal{K}|$. However, since the estimation and detection are performed in the presence of noise, the number of $P$ needs to be increased in general to enhance the sensitivity of the receiver. This gives the option of trading off accuracy with { storage cost and} computational complexity, by adjusting the number of samples $P$ to process between ${P_{\rm min}}  \leq {P}  \leq {P_{\rm MF}}$.



%

Note that for different schemes, we need further processing to produce final decisions. This last step is different for different receivers. For example, the MF scheme has very simple post-processing at the cost of handling exhaustive MF samples, while the C-SA scheme can tune the number of measurements to handle less data by spending a higher premium for sparsity recovery algorithms. Therefore, we discuss this trade-off specifically in detail in Section \ref{complexity}, with the MF being a benchmark for our comparison.

\subsection{CMS Observation Model}\label{orig_link_acq}
Similar to \cite{fletcher2009off,applebaum2011asynchronous}, we follow the analog description of \eqref{orig_sig} but discretize the parameters as in the MF approach. The analog domain derivation is mostly inspired by the FRI literature, but it is not tied to estimating continuous parameters here in this paper. For notational convenience, we introduce the triple-index coefficient
\begin{align}\label{a_ikq}
	\alpha_{i,k,q} = \sum_{ j\in\mathcal{I}}\sum_{r=1}^{R} h_{j,r}\delta[i-j]\delta[k-k_{j,r}]\delta[q-q_{j,r}],
\end{align}
for $k\in\mathcal{K},q\in\mathcal{Q}$ as an indicator of whether the $i$th user is transmitting and whether there exists a link at a certain delay $\tau=q\Delta\tau$ with a certain carrier offset $\omega=k\Delta\omega$ in the window. Note that $\alpha_{i,k,q}=0$ except when $k=k_{i,r}$ and $q=q_{i,r}$ for $i\in\mathcal{I}$. Denoting each {\it MF template} by
\begin{align}\label{templates}
	\phi_{i,k,q}(t) \triangleq \phi_i(t-q\Delta\tau)e^{\mathrm{i}k\Delta\omega t},
\end{align}
the signal in \eqref{orig_sig} can be approximately expressed as
\begin{align}\label{sparse_sig}
	x(t) &=\sum_{i=1}^{I}\sum_{k\in\mathcal{K}} \sum_{q\in\mathcal{Q}}\alpha_{i,k,q}e^{\mathrm{i}k\Delta\omega \ell D}~\phi_{i,k,q}(t-\ell D)+ v(t).
\end{align}
Clearly, $x(t)$ has at most $|\mathcal{I}|R$ active components due to the sparsity of $\alpha_{i,k,q}$.
To facilitate notations in our derivations, we introduce the triplet index $(i,k,q)$ and define the length-$I|\mathcal{K}||\mathcal{Q}|$  {\it link vector} ${ { \bdsb{\alpha}[\ell]}}$ at the $\ell$th shift as
\begin{align}\label{index_mu}
	\Big[{ { \bdsb{\alpha}[\ell]}}\Big]_{(i,k,q)} &\triangleq \Big[{ { \bdsb{\alpha}[\ell]}}\Big]_{(i-1)|\mathcal{K}||\mathcal{Q}| + (k+K-1)|\mathcal{Q}|+q} \\
	&= \alpha_{i,k,q}
\end{align}
for any $i=1,\cdots,I$, $k\in\mathcal{K}$ and $q\in\mathcal{Q}$. We define the associated {\it delay-Doppler set} for the $i$th user at the $\ell$th shift
\begin{align}\label{A_i}
	{ \mathcal{A}_i^{(\ell)}} &\triangleq \left\{(k,q) : |\alpha_{i,k,q}| > 0, k\in\mathcal{K},q\in\mathcal{Q}\right\},
\end{align}
from which we extract the delays $\tau_{i,r}=q_{i,r}\Delta\tau$ and carrier offsets $\omega_{i,r}=k_{i,r}\Delta\omega$ for the active users $i\in\mathcal{I}$ if ${ \mathcal{A}_i^{(\ell)}}\neq \varnothing$.

In this work, we consider the use of sampling kernels  $\psi_p(t)$ that are linear combinations of all the MF templates \cite{eldar2009compressed}. The following theorem specifies the observed samples from the CMS architecture in Fig. \ref{CS} in relation to the link vector ${ { \bdsb{\alpha}[\ell]}}$.
\begin{thm}\label{observation_model}
Suppose that we choose sampling kernels $\{\psi_p(t)\}_{p=1}^{P}$ as linear combinations of the MF templates
	\begin{align}\label{sampling_kernels}
		\psi_p(t) = \sum_{i=1}^{I}\sum_{k\in\mathcal{K}}\sum_{q\in\mathcal{Q}} b_{p,(i,k,q)} \phi_{i,k,q}(t),~~ p=1,\cdots,P.
	\end{align}
The length-$P$ sample vector $\mathbf{c}[n]\triangleq [c_1[n],\cdots,c_P[n]]^T$ taken at shift $t=nD$, can then be expressed as
	\begin{align}\label{c_1}
		\mathbf{c}[n] = \mathbf{B}\mathbf{M}_{\phi\phi}[n-\ell] \bdsb{\Gamma}[\ell] { { \bdsb{\alpha}[\ell]}} +\bdsb{\nu}[n],
	\end{align}
where $\bdsb{\alpha}[\ell]$ is the link vector at the $\ell$th shift and\\
$1)$ $\mathbf{B}$ is a $P\times I|\mathcal{K}||\mathcal{Q}|$ matrix with $\left[\mathbf{B}\right]_{p, (i,k,q)} \triangleq b_{p,(i,k,q)}$;\\
$2)$ $\mathbf{M}_{\phi\phi}[n-\ell]$ is an $I|\mathcal{K}||\mathcal{Q}|\times I|\mathcal{K}||\mathcal{Q}|$ matrix with
		\begin{align*}
			& \Big[\mathbf{M}_{\phi\phi}[n-\ell]\Big]_{(i',k',q'),(i,k,q)} \triangleq R_{\phi_{i',k',q'}\phi_{i,k,q}}[(n-\ell)D],
		\end{align*}
		where
		\begin{align*}
			&R_{\phi_{i',k',q'}\phi_{i,k,q}}[\Delta t]\\
			&~~~~~~~~~~~~~~~ = e^{\mathrm{i}k\Delta\omega(\Delta t-q'\Delta\tau)} R_{\phi_{i'},\phi_i}^{(k-k')}\left[(q'-q)\Delta\tau + \Delta t\right],
		\end{align*}		
	       and $R_{\phi_{i'},\phi_i}^{(k-k')}(\cdot)$ is the ambiguity function
		\begin{align}\label{ambiguity}
			R_{\phi_{i'},\phi_i}^{(k-k')}(\cdot)
			&=\int \phi_{i'}^\ast(t)\phi_i(t-\cdot)e^{\mathrm{i}(k-k')\Delta\omega t}\mathrm{d}t.
		\end{align}
$3)$ $\bdsb{\Gamma}[\ell] \triangleq \mathbf{I}_I\otimes\mathbf{E}[\ell]\otimes \mathbf{I}_{|\mathcal{Q}|}$ with 
			\begin{align}
				\mathbf{E}[\ell] \triangleq \mathrm{diag}[e^{\mathrm{i}K\Delta\omega\ell D},\cdots,e^{-\mathrm{i}K\Delta\omega \ell D}];
			\end{align}
$4)$ $\bdsb{\nu}[n]\triangleq [\nu_1[n],\cdots,\nu_P[n]]^T$ is the filtered Gaussian noise vector with zero mean and covariance
			\begin{align}
				\mathbf{R}_{vv}=\sigma^2\mathbf{R}_{\psi\psi}~\textrm{with}~\mathbf{R}_{\psi\psi} = \mathbf{B}\mathbf{M}_{\phi\phi}[0]\mathbf{B}^H.
			\end{align}
\end{thm}
\begin{proof}
See Appendix \ref{theorem_1}.
\end{proof}
The freedom in choosing $\mathbf{B}$ allows us to optimize acquisition performance. Before discussing the details of optimization in Section \ref{optimization_B}, we further simplify the model in Theorem \ref{observation_model}.

\subsection{CMS Sequential Acquisition Model}\label{reformulated_problem}
Theorem \ref{observation_model} describes the general model of the samples $\mathbf{c}[n]$ obtained in the $n$th shift with respect to the link vector ${ { \bdsb{\alpha}[\ell]}}$. However, the exact shift $\ell$ is unknown to the receiver. As mentioned earlier, determining the exact shift is not necessary to recover the link parameters, as long as the shift is properly aligned with the signal and produces a positive detection maximizing the likelihood ratio. In the following, we transform the observation model $\mathbf{c}[n]$ in Theorem \ref{observation_model} to an equivalent model. The equivalent model is stated with respect to a modified link vector ${ \bdsb{\alpha}[n]}$ { at the $n$th shift}, which contains entries that are shifted with the relative placement of $(n-\ell)$ in relation to ${ { \bdsb{\alpha}[\ell]}}$. The reason for this is that we can use a time-invariant system matrix instead of a time-variant one $\mathbf{M}_{\phi\phi}[n-\ell]$ for the purpose of sequential detection.
\begin{thm}\label{observation_model_reform}
	Let $D=N\Delta\tau$ for some integer $N\in\mathbb{Z}$. The outputs $\mathbf{c}[n]$ of the compressive samplers can be re-written as
	\begin{align}\label{thm_2}
		\mathbf{c}[n] =
		\mathbf{B}\mathbf{M}{ \bdsb{\Gamma}[n]\bdsb{\alpha}[n]} + \bdsb{\nu}[n],
	\end{align}
	 where $\mathbf{M}\triangleq\mathbf{M}_{\phi\phi}[0]$ and $\bdsb{\alpha}[n]$ is the $n$th shift link vector
	 \begin{align}
	 	\Big[{ \bdsb{\alpha}[n]}\Big]_{(i,k,q)}
		 \triangleq \alpha_{i,k,q+(n-\ell)N}. 
	\end{align}
\end{thm}
\begin{proof}
	See Appendix \ref{theorem_2}.
\end{proof}
{
\begin{cor}
Let the {\it delay-Doppler sets} at the $n$th shift be
\begin{align}\label{S_i}
	{ \mathcal{A}_{i}^{(n)}} &\triangleq \left\{(k,q) : \left|\Big[{ \bdsb{\alpha}[n]}\Big]_{(i,k,q)}\right| \neq 0, k\in\mathcal{K},q\in\mathcal{Q}\right\}
\end{align}
for $i=1,\cdots,I$. Then for any $(k,q)\in{ \mathcal{A}_i^{(\ell)}}$ at the $\ell$th shift \eqref{A_i}, we have $(k, q+(\ell-n)N)\in{ \mathcal{A}_{i}^{(n)}}$ at the $n$th shift.
\end{cor}
}
Using the modified sets ${ \mathcal{A}_{i}^{(n)}}$, the number of delay-Doppler pairs included at the $n$th shift equals $\sum_{i=1}^{I}|{ \mathcal{A}_{i}^{(n)}}|$. It is obvious that $|{ \mathcal{A}_{i}^{(n)}}|\leq |{ \mathcal{A}_i^{(\ell)}}|$ for any $i$. At the $n$th shift, if a positive det{ection is declared and $|{ \mathcal{A}_{i}^{(n)}}|= |{ \mathcal{A}_i^{(\ell)}}|$ for all $i$, then the modified link vector $\bdsb{\alpha}[n]$ at the $n$th shift carries equivalent link information as the link vector $\bdsb{\alpha}[\ell]$ at the $\ell$th shift.
Therefore, we use the model in Theorem \ref{observation_model_reform} for our design and re-state the goal of {\it link acquisition} as
\begin{enumerate}
	\item locating the MLR shift $\ell_\star$;
	\item identifying the set of active users $\widehat{\mathcal{I}}$ indicated by the {\it delay-Doppler set} $\mathcal{A}_i^{(\ell_\star)}\neq \varnothing$;
	\item resolving the delay-Doppler pairs in the $\ell_\star$th window $\mathcal{A}_i^{(\ell_\star)}\subseteq \mathcal{K}\times\mathcal{Q}$ for $i\in\widehat{\mathcal{I}}$.
\end{enumerate}

For better representation and comparison of the individual support set ${ \mathcal{A}_{i}^{(n)}}$ in relation to the original support set ${ \mathcal{A}_i^{(\ell)}}$, we introduce the full {\it { user-delay-Doppler sets}} for the link vectors ${ \bdsb{\alpha}[n]}$ and ${ \bdsb{\alpha}[\ell]}$ respectively
\begin{align}
	{ \mathcal{A}_n} &\triangleq \left\{(i,k,q) : ~(k,q)\in{ \mathcal{A}_{i}^{(n)}}, i\in\mathcal{I}\right\},\\
	{ \mathcal{A}_\ell} &\triangleq \left\{(i,k,q) : ~(k,q)\in{ \mathcal{A}_i^{(\ell)}}, i\in\mathcal{I}\right\}.
\end{align}
In the following sections, { we express the link vector explicitly with respect to the full {\it user-delay-Doppler set} ${ \mathcal{A}_n}$ and combine the phase rotation matrix $\bdsb{\Gamma}[n]$ at the $n$th shift as
\begin{align}
	{ \bdsb{\beta}_{{ \mathcal{A}_n}}\triangleq \bdsb{\Gamma}[n]\bdsb{\alpha}[n]}.
\end{align}
We call ${ \bdsb{\beta}_{{ \mathcal{A}_n}}=\left[\cdots,\beta_{i,k,q},\cdots\right]^T}$ the {\it modified link vector} and note that it is also a $|\mathcal{I}|R$-sparse vector.}


%
\section{Compressive Sequential Link Acquisition with Sparsity Regularization}\label{SR_GLRT}
We now develop an SR-LRT detection algorithm that tackles the link acquisition problem exploiting the compressive observation model given in Theorem \ref{observation_model_reform}. Link acquisition attempts to discriminate the true pattern ${ { \mathcal{A}_n}}$ against all possible patterns ${ \mathcal{S}_n}\neq { { \mathcal{A}_n}}$ at every shift $t=nD$ as a hypothesis test
\begin{align}\label{thm_2_reform}
	\mathcal{H}_{{ \mathcal{S}_n}} : &~~~
        \mathbf{c}[n] = \mathbf{B}\mathbf{M}{ \bdsb{\beta}_{\mathcal{S}_n}} + \mathbf{v}[n]
\end{align}
over all possible ${ \mathcal{S}_n}$ at every shift $t=nD$. Note that the signal presence detection is incorporated in this test by choosing $\mathcal{S}_n=\varnothing$ to be the null hypothesis. Given a specific set $\mathcal{S}_n$ for each possible $\bdsb{\beta}_{\mathcal{S}_n}$, the amplitudes of $\bdsb{\beta}_{\mathcal{S}_n}$ and the noise variance $\sigma^2$ are unknown and treated as nuisance parameters. The link acquisition is thus to detect the full {\it user-delay-Doppler set} ${ \mathcal{S}_n}$ for all possible $\mathcal{H}_{{ \mathcal{S}_n}}$ with ${ \bdsb{\beta}_{\mathcal{S}_n}}$ and the noise level $\sigma^2$ being nuisance in each shift $n$. Following the GLRT rationale, the test consists in finding the set ${ \mathcal{S}_n}$ maximizing
\begin{align}\label{P_H}
	&\mathbb{P}\left(\mathcal{H}_{{ \mathcal{S}_n}}|{ \bdsb{\beta}_{\mathcal{S}_n}},\sigma^2\right)\\
	&=\frac{1}{\pi^P\sigma^{2P}|\mathbf{R}_{\psi\psi}|}
	 \exp\left(-\frac{\left\|\mathbf{c}[n]-\mathbf{B}\mathbf{M}{ \bdsb{\beta}_{\mathcal{S}_n}}\right\|^2_{\mathbf{R}_{\psi\psi}^{-1}}}{\sigma^2}\right)
\end{align}
in the presence of unknown parameters ${ \bdsb{\beta}_{\mathcal{S}_n}}$ and $\sigma^2$.

Note that when $\mathbf{B}=\mathbf{I}$ is identity, the samples $\mathbf{c}[n]$ are equivalent to the outputs of the MF approach. This implies that the samples $\mathbf{c}[n]$ obtained in the CMS architecture, using only $P$ sampling kernels in \eqref{sampling_kernels}, are equivalent to a linearly compressed version of the exhaustive MF output, even though they are obtained directly from the A/D architecture instead of using an exhaustive MF filterbank followed by a linear compressor $\mathbf{B}$ as in \cite{wang2008subspace,haupt2007compressive}, which would be much more complex. On the other hand, the difference in post-processing between the MF and CMS architecture is that MF allows to simply pick the hypothesis corresponding to the largest magnitude in the output $\mathbf{c}[n]$ as the detection result, while a more sophisticated detection scheme is necessary if we use the compressive samples $\mathbf{c}[n]$.


\subsection{Sequential Estimation for Link Acquisition}

The GLRT involved in the link acquisition requires estimating ${ \bdsb{\beta}_{\mathcal{S}_n}}$ and $\sigma^2$ for every possible ${ \mathcal{S}_n}$ at every shift $t=nD$. 
For every hypothesis ${ \mathcal{S}_n}\neq\varnothing$, the estimate of $\bdsb{\beta}_{\mathcal{S}_n}$ under colored Gaussian noise $\bdsb{\nu}[n]$ with covariance $\mathbf{R}_{vv}=\sigma^2\mathbf{R}_{\psi\psi}$ is then obtained as
\begin{align}\label{fitting}
	{ \bdsb{\widehat{\beta}}_{\mathcal{S}_n}} \triangleq \arg &~\underset{{ \bdsb{\beta}_{\mathcal{S}_n}}}{\min}~ \left\|\mathbf{c}[n]-\mathbf{B}\mathbf{M}{ \bdsb{\beta}_{\mathcal{S}_n}}\right\|^2_{\mathbf{R}_{\psi\psi}^{-1}}.
\end{align}
The ``hat" notation $\widehat{(\cdot)}$ on the vector ${ \bdsb{\beta}_{\mathcal{S}_n}}$ refers to the estimates of the amplitudes on the support ${ \mathcal{S}_n}$. The total number of such estimates scales with the number of hypothesis which in this case is $2^{I|\mathcal{K}||\mathcal{Q}|}$, resulting in an NP-hard combinatorial estimation problem. Instead, we obtain the estimates using a sparse approach similar to \cite{fletcher2009off,applebaum2011asynchronous} for the GLRT. Specifically, we solve the combinatorial problem in a ``soft" fashion at every shift $t=nD$ similar to \cite{applebaum2011asynchronous}
\begin{align}\label{sparse_est}
	\bdsb{\widehat{\beta}}\triangleq \arg &~\underset{\bdsb{\beta}}{\min}~ \left\|\mathbf{c}[n]-\mathbf{B}\mathbf{M}\bdsb{\beta}\right\|^2_{\mathbf{R}_{\psi\psi}^{-1}} + \lambda\cdot f(\bdsb{\beta}),
\end{align}
where $\lambda$ is some regularization parameter and $f(\bdsb{\beta}) = \left\|\bdsb{\beta}\right\|_0$ or $f(\bdsb{\beta})=\left\|\bdsb{\beta}\right\|_1$ are the sparsity regularization constraint. If the $\|\cdot\|_0$ constraint is imposed, the problem is approximately solved via greedy methods such as orthogonal matching pursuit (OMP) \cite{pati1993orthogonal}. When $\|\cdot\|_1$ norm is used, this problem can be solved via convex programs \cite{candes2006stable}. Generally speaking, as discussed in \cite{applebaum2011asynchronous} and \cite{candes2006stable}, the required number of samples $P$ for sparse recovery in the noiseless case scales logarithmically with the length-$I|\mathcal{K}||\mathcal{Q}|$ link vector $P \propto |\mathcal{I}|R \log I|\mathcal{K}||\mathcal{Q}|$, which increases if discretizations are made finer.

From the solution of \eqref{sparse_est}, we extract the full {\it user-delay-Doppler set} ${ \widehat{\mathcal{A}}_n}$ from the soft estimate $\bdsb{\widehat{\beta}}$, which in turns gives the estimated user set $\widehat{\mathcal{I}}$ and the estimated {\it delay-Doppler set} $\widehat{\mathcal{A}}_i^{(n)}$ for each user $i\in\widehat{\mathcal{I}}$
\begin{align}
    \mathcal{E}\left(\bdsb{\widehat{\beta}}\right) = \widehat{\mathcal{A}}_n \quad
    \Longrightarrow
    \quad
    \widehat{\mathcal{I}}~\textrm{and}~\left\{\widehat{\mathcal{A}}_i^{(n)}\right\}_{i\in\widehat{\mathcal{I}}}.
\end{align}
Here $\mathcal{E}(\cdot)$ is the extraction mapping from the soft estimate to the estimated {\it user-delay-Doppler} set. The extraction method is explained in Section \ref{UDD_extraction}.

With the estimated set of active users $\widehat{\mathcal{I}}$ and individual {\it delay-Doppler set} ${ \widehat{\mathcal{A}}_i^{(n)}}$, we have the truncated estimate of the link vector ${ \bdsb{\widehat{\beta}}_{\widehat{\mathcal{A}}_n}}$ and the estimated noise variance 
\begin{align}\label{sigma}
	{ \widehat{\sigma}_{\widehat{\mathcal{A}}_n}^2} &= \Big\|\mathbf{c}[n]-\mathbf{B}\mathbf{M}{ \bdsb{\widehat{\beta}}_{ \widehat{\mathcal{A}}_n}}\Big\|_{\mathbf{R}_{\psi\psi}^{-1}}^2/P.
\end{align}
{ Since the formulation in \eqref{sparse_est} is no longer maximum likelihood due to the sparsity regularization, we call it the Sparsity-Regularized Likelihood Ratio Test (SR-LRT).}

\subsection{{ User-Delay-Doppler} Set Extraction $\mathcal{E}\left(\bdsb{\widehat{\beta}}\right)$}\label{UDD_extraction}
Given the soft estimate ${ \bdsb{\widehat{\beta}}}$ in \eqref{sparse_est} at every shift $t=nD$, the estimated user-delay-Doppler set ${ \widehat{\mathcal{A}}_n}$ is extracted depending on the application scenarios below.

\subsubsection{\sl Unknown, random number of active users $\mathcal{I}$}\label{unknown_user_case}
In random access communications the receiver has no knowledge of who is active, nor any expectation on the number of components it is likely to detect.
Using this soft estimate ${ \bdsb{\widehat{\beta}}}$, we identify the active users as
\begin{align}\label{G_i}
	\widehat{\mathcal{I}} &\triangleq \left\{ i : \underset{k,q}{\max}~\left|\widehat{\beta}_{i,k,q}\right|^2\geq \rho_i, k\in\mathcal{K},q\in\mathcal{Q}\right\}
\end{align}
where $\rho_i$ is a chosen threshold for that specific user to be considered present, usually set as a fraction of the magnitude of the amplitudes in $\widehat{\bdsb{\beta}}$. Then for each detected active user $i\in\widehat{\mathcal{I}}$, we take $R$ strongest paths in $\widehat{\beta}_{i,k,q}$ with respect to $k\in\mathcal{K}$ and $\mathcal{Q}$ to be the active set ${ \widehat{\mathcal{A}}_i^{(n)}}$ for each user $i\in\widehat{\mathcal{I}}$.

\subsubsection{\sl Partial knowledge on active users $\mathcal{I}$}\label{partial_user_case}
This scenario corresponds to environments where all users are active, however only a certain subset is likely to be detectable by the receiver. GPS receivers are an example. Specifically, there are a total of $I=24$ quasi-stationary GPS satellites moving around the earth and the active satellites in the field-of-view of a specific GPS receiver are unknown. However, the GPS receiver is informed that at any point in space there should be  $|\mathcal{I}|=4$ strongest signals from satellites, and it attempts to find such signals, along with their delay-Doppler parameters for triangularization. In this case a positive detection corresponds to having at least four components detected and we can interpret this case as fixing $|\mathcal{I}|$ for the receiver detection.  In general, we identify the users $i\in\widehat{\mathcal{I}}$ as those with the $|\mathcal{I}|$ strongest amplitudes $|\widehat{\beta}_{i,k,q}|$ with respect to $i=1,\cdots,I$ in ${ \bdsb{\widehat{\beta}}}$. Then we take $R$ strongest paths in $|\widehat{\beta}_{i,k,q}|$ with respect to $k$ and $q$ to be the active set ${ \widehat{\mathcal{A}}_i^{(n)}}$ for each user $i\in\widehat{\mathcal{I}}$.

\subsubsection{\sl Known active users $\mathcal{I}$} This scenario includes multi-antenna and cooperative transmission systems, where the receiver is aware of the active sources, i.e., $\mathcal{I}$ is known. This case is trivial because we do not need to identify the active users. The active set ${ \widehat{\mathcal{A}}_i^{(n)}}$ for each user $i\in \mathcal{I}$ is formed by picking the $R$ strongest components in $|\widehat{\beta}_{i,k,q}|$ with respect to $k\in\mathcal{K}$ and $\mathcal{Q}$.

In fact, there are many applications that require multiuser acquisition. In a cellular environment the mobile may try to detect the presence of base-stations, and typically there is a handful of them whose synchronization signals are not completely overwhelmed by noise. For example in LTE systems there are X distinct synchronization sequences that are used to signal the presence of a tower to mobile users. In the example we used the setting $I_{\max}=4$ and $R=2$ as an easy example resembling the application of GPS, which captures $4$ satellites for triangularization that travels mainly with a line-of-sight component. In the revised manuscript Section V, we describe the GPS example in more detail. In cellular systems, the sets of active users can reach much larger numbers (hundreds), but the link acquisition is typically done to acquire base-station signals in the downlink, with a handful of dominating signals from the nearest base-stations. In the uplink, more often, dedicated random access control channels are used to detect the presence of a subscriber in the cell, and the collision rate in them is usually small. Hence, the example of  $I_{\max}=4$ is not so far off, as far as link acquisition is concerned.

\subsection{Sequential Detection for Link Acquisition}
Substituting  ${ \bdsb{\widehat{\beta}}}_{{ \widehat{\mathcal{A}}_n}}$ and ${ \widehat{\sigma}_{ \widehat{\mathcal{A}}_n}}^2$ back to \eqref{P_H}, the generalized likelihood ratio can be computed as
\begin{align}\label{eta_n}
	{ \eta_{\textrm{\tiny C-SA}}(n)} &\triangleq \frac{\mathbb{P}\left(\mathcal{H}_{{ \widehat{\mathcal{A}}_n}}|{ \bdsb{\widehat{\beta}}}_{{ \widehat{\mathcal{A}}_n}},\widehat{\sigma}_{{ \widehat{\mathcal{A}}_n}}^2\right)}{\mathbb{P}\left(\mathcal{H}_{\varnothing}|\widehat{\sigma}_{\varnothing}^2\right)}\\
	 &=\frac{\|\mathbf{c}[n]\|^{2P}_{\mathbf{R}_{\psi\psi}^{-1}}}{\left\|\mathbf{c}[n]-\mathbf{B}\mathbf{M}{ \bdsb{\widehat{\beta}}}_{{ \widehat{\mathcal{A}}_n}}\right\|^{2P}_{\mathbf{R}_{\psi\psi}^{-1}}}> \eta_0,
\end{align}
which indicates the presence of the signal if ${ \eta_{\textrm{\tiny C-SA}}(n)} > \eta_0$ so that the receiver knows that certain signal components are captured in the observation. Denote the first window that passes the above test as 
\begin{align}
N_\eta\triangleq \min\left\{\underset{n}{\arg}~{ \eta_{\textrm{\tiny C-SA}}(n)}\geq \eta_0\right\}. 
\end{align}
As mentioned in \eqref{GLRT_best_window}, the MLR window is located as the window that maximizes the likelihood ratio
\begin{align}
	\ell_\star = \arg\underset{n}{\max}~~{ \eta_{\textrm{\tiny C-SA}}(n)},\quad n=N_\eta,\cdots,N_\eta+N_0.
\end{align}
Accordingly, from the link vector ${ \bdsb{\widehat{\beta}}_{\widehat{\mathcal{A}}_{\ell_\star}}}$ in the $\ell_\star$th window, we can extract the delay-Doppler pairs
\begin{align}\label{asyn:eq8}
	\widehat{\tau}_{i,r} = q_{i,r}\Delta\tau, ~
	\widehat{\omega}_{i,r} = k_{i,r}\Delta\omega, ~ (k,q)\in\widehat{\mathcal{A}}_i^{(\ell_\star)},~i\in\widehat{\mathcal{I}}.
\end{align}

\section{Optimization of Compressive Samplers}\label{optimization_B}
The link acquisition performance depends on the ability of the SR-LRT to differentiate between different hypotheses $\mathcal{H}_{\mathcal{S}_n}$. In this section, we seek a criterion to optimize the sampling kernels $\{\psi_p(t)\}_{p=1}^{P}$ by designing the matrix $\mathbf{B}$. The metric we maximize is the weighted average of the Kullback-Leibler (KL) distances between any $\mathcal{H}_{\mathcal{S}_n}$ in \eqref{thm_2_reform}. Since every possible pattern for $\mathcal{S}_n$ is independent of $n$, here we omit the subscript for convenience. In choosing the KL distance we are motivated by the Chernoff-Stein's lemma \cite{cover1991elements}, whose statement indicates that the probability of confusing $\mathcal{H}_{\mathcal{S}}$ and $\mathcal{H}_{\mathcal{S}'}$ decreases exponentially with the pair-wise KL distance between them. As we point out later in this section, if the noise is Gaussian and the weights are chosen appropriately, then the weighted average KL distance of all the {\it pair-wise KL} distances has the same expression as the Chernoff information under the Bayesian detection framework, implying that the average KL distance is an effective measure in evaluating detection performance. Being consistent with our system model and detection formulation, we proceed with our analysis using the average KL distance with some pre-defined weights.

Introducing $\mathcal{G}\left(\mathbf{B}\right)\triangleq\mathbf{M}^H\mathbf{B}^H\left(\mathbf{B}\mathbf{M}\mathbf{B}^H\right)^{-1}\mathbf{B}\mathbf{M}$, the {\it pair-wise KL} distance between any $\mathcal{H}_{\mathcal{S}}$ and $\mathcal{H}_{\mathcal{S}'}$ is given by \cite{bajovic2009robust}
\begin{align}\label{pair-wise_KL}
	 &\mathbb{D}\left(\mathcal{H}_{\mathcal{S}}\|\mathcal{H}_{\mathcal{S}'}\right)
	 =  \frac{\left(\bdsb{\beta}_{\mathcal{S}}-\bdsb{\beta}_{\mathcal{S}'}\right)^H\mathcal{G}\left(\mathbf{B}\right)\left(\bdsb{\beta}_{\mathcal{S}}-\bdsb{\beta}_{\mathcal{S}'}\right)}{\sigma^2}.
\end{align}
If the pair-wise KL distance is zero, then the two hypotheses $\mathcal{H}_{\mathcal{S}}$ and $\mathcal{H}_{\mathcal{S}'}$ are indistinguishable for that particular pairs of $\mathcal{S}$ and $\mathcal{S}'$. A non-zero pair-wise KL distance between arbitrary pair of $\mathcal{S}$ and $\mathcal{S}'$ with $|\mathcal{S}|, |\mathcal{S}'|\leq s$ requires $\mathrm{spark}\left[\mathcal{G}\left(\mathbf{B}\right)\right]\geq 2s$, where $\mathrm{spark}[\cdot]$ is the kruskal rank of a matrix. We note that the average KL distance metric defined here does not automatically ensure that $\mathbb{D}\left(\mathcal{H}_{\mathcal{S}}\|\mathcal{H}_{\mathcal{S}'}\right) >0$ for all $\mathcal{S}\neq\mathcal{S}'$.

 To define the average KL distance, we associate each distinct pair of supports $\mathcal{S}$ and $\mathcal{S}'$ with the weight $\gamma_{\mathcal{S},\mathcal{S}'}$. Furthermore, we associate the nuisance amplitudes in $\bdsb{\beta}_{\mathcal{S}}$ a multidimensional continuous weighting function  $P(\bdsb{\beta}_{\mathcal{S}})$ for any $\mathcal{S}$. Under these assumptions, the weighted average of all {\it pair-wise KL} distances is defined as
\begin{align}\label{avg_KL}
	\overline{\mathbb{D}}
	&=
	\sum_{\mathcal{S}}\sum_{\mathcal{S}'\neq \mathcal{S}}
	\gamma_{\mathcal{S},\mathcal{S}'}\iint P(\bdsb{\beta}_{\mathcal{S}})P(\bdsb{\beta}_{\mathcal{S}'})
	 \mathbb{D}\left(\mathcal{H}_{\mathcal{S}}\|\mathcal{H}_{\mathcal{S}'}\right)\mathrm{d}\bdsb{\beta}_{\mathcal{S}}\mathrm{d}\bdsb{\beta}_{\mathcal{S}'}.
\end{align}

\begin{prop}
	Given a set of normalized weights $\gamma_{\mathcal{S},\mathcal{S}'}$ for every distinct pair $\mathcal{S}\neq\mathcal{S}'$, and a continuous weighting function $P(\bdsb{\beta}_{\mathcal{S}})=\prod_{(i,k,q)\in\mathcal{S}}P(\beta_{i,k,q})$ over the nuisance amplitudes with $\int \bdsb{\beta}_{\mathcal{S}}P(\bdsb{\beta}_{\mathcal{S}})\mathrm{d}\bdsb{\beta}_{\mathcal{S}}= \mathbf{0}$ and $\int |\beta_{i,k,q}|^2 P(\beta_{i,k,q})\mathrm{d}\beta_{i,k,q}= \sigma_\beta^2$, the average KL distance $\overline{\mathbb{D}}$ is equal to
	\begin{align}\label{avg_D}
		 \overline{\mathbb{D}}=\frac{\sigma_\beta^2}{\sigma^2}\mathrm{Tr}\left[\mathbf{M}^H\mathbf{B}^H\left(\mathbf{B}\mathbf{M}\mathbf{B}^H\right)^{-1}\mathbf{B}\mathbf{M}\right].
	\end{align}
\end{prop}
\begin{proof}
See Appendix \ref{KL_max}.
\end{proof}
{ The way we choose the weights $\gamma_{\mathcal{S},\mathcal{S}'}$ and weighting functions $P(\bdsb{\beta}_{\mathcal{S}})$ is equivalent to assuming a uniform distribution on the users, delays and Dopplers together with i.i.d. Gaussian priors on the amplitudes in $\bdsb{\beta}_{\mathcal{S}}$ in a Bayesian detection framework. According to \cite{bajovic2009robust}, the average KL distance obtained in \eqref{avg_D} has the same expression as the Chernoff information, which determines the Bayesian detection error exponent. Therefore, the average KL distance measure in a sense maximizes the error exponent in the exponential decay on the Bayesian detection error performance (or miss detection performance under the Neyman-Pearson detection framework).}

We note that it is possible that specific choices of the number of samplers $P$ and the dictionary $\{\phi_{i,k,q}(t)\}_{i=1,\cdots,I}^{k\in\mathcal{K},q\in\mathcal{Q}}$ (i.e., the Gram matrix $\mathbf{M}$) lead to indistinguishable sparsity patterns \cite{candes2006stable} such that $\mathrm{spark}\left(\mathcal{G}\left(\mathbf{B}\right)\right)\leq|\mathcal{S}|+|\mathcal{S}'|$. In other words, the design of $\mathbf{B}$ cannot cure intrinsic problems caused by the choice of $P$ or the Gram matrix $\mathbf{M}$, which are given in the optimization. The intrinsic problems caused by the Gram matrix $\mathbf{M}$ in communications is typically handled by optimizing the transmit sequences $\phi_i(t)$ irrespective of the receiver such that $\mathbf{M}$ becomes diagonally dominated. This implies a well localized ambiguity function for each of the $\phi_i(t)$ and low cross-correlation between $\phi_i(t)$'s with different delays and Doppler. Gold sequences used in GPS and M-sequences used in spread spectrum communications, for example, are known to have good properties in this regard. This is a well investigated problem \cite{klauder1960design} that we do not aim to cover here.

Given $P$ and $\mathbf{M}$, we propose an optimal $\mathbf{B}$ that maximizes the average KL distance $\overline{\mathbb{D}}$ if there is a unique solution to the optimization; when there are multiple solutions that yield identical average KL distance $\overline{\mathbb{D}}$, we further choose in the feasible set the matrix $\mathbf{B}$ that gives the least occurrence of events $\mathbb{D}\left(\mathcal{H}_{\mathcal{S}}\|\mathcal{H}_{\mathcal{S}'}\right) =0$. We use the results in the following lemma for our optimization.

\begin{lem}\label{ratio_trace}
{\bf (Ratio Trace Maximization \cite{duda1995pattern})} Given a pair of $L\times L$ positive semi-definite matrices $(\mathbf{S},\mathbf{G})$ and an $L\times P$ full column rank matrix $\mathbf{W}$, the ratio trace problem is
\begin{align}
	\mathbf{W}^{\rm opt} = \arg\underset{\mathbf{W}}{\max}~\mathrm{Tr}\left[\left(\mathbf{W}^H\mathbf{S}\mathbf{W}\right)^{-1}\mathbf{W}^H\mathbf{G}\mathbf{W}\right].
\end{align}
The optimal $\mathbf{W}^{\rm opt} = [\mathbf{w}_1^{\rm opt},\cdots,\mathbf{w}_P^{\rm opt}]$ is given by the generalized eigenvectors $\{\mathbf{w}_p^{\rm opt}\}_{p=1}^P$ corresponding to $P$ largest generalized eigenvalues of the pair $(\mathbf{S},\mathbf{G})$ with $P\leq \mathrm{rank}[\mathbf{S}]$.
\end{lem}
The optimal $\mathbf{B}$ is identified in the following theorem\footnote{Note that the computation of the weights is done offline, and does not add complexity to the online processing.}:

\begin{thm}\label{KL_opt}
	Let the eigendecomposition be $\mathbf{M}=\mathbf{U}\bdsb{\Sigma}\mathbf{U}^H$ where $\bdsb{\Sigma}=\mathrm{diag}[\sigma_1,\cdots,\sigma_{I|\mathcal{K}||\mathcal{Q}|}]$ is the eigenvalue matrix in descending order and $\mathbf{U}$ is the eigenvector matrix of $\mathbf{M}$. Denote the set of $P\leq \mathrm{rank}[\mathbf{M}]$ principal eigenvectors
	\begin{align}
		\mathcal{U}&=\Big\{\mathbf{U}_P = [\mathbf{u}_1,\cdots,\mathbf{u}_P] :\\
        &~~~~~~~\mathbf{M}=\mathbf{U}\bdsb{\Sigma}\mathbf{U}^H, \mathbf{U}= \left[\mathbf{u}_1,\cdots,\mathbf{u}_{I|\mathcal{K}||\mathcal{Q}|}\right]\Big\}.
	\end{align}
	Let $\bdsb{\Xi}_P$ be an arbitrary non-singular $P\times P$ matrix. When $\mathbf{U}_P$ is unique, the matrix $\mathbf{B}_\star = \bdsb{\Xi}_P\mathbf{U}_P^H$ is chosen to maximize the average KL distance $\overline{\mathbb{D}}$ in \eqref{avg_KL}. When $\mathbf{U}_P$ is not unique, we choose $\mathbf{\widehat{U}}_P = \max_{\mathbf{U}_P\in\mathcal{U}}~\mathrm{spark}\left(\mathbf{U}_P^H\right)$
to maximize the average KL distance and minimize the occurrence of events $\mathbb{D}\left(\mathcal{H}_{\mathcal{S}}\|\mathcal{H}_{\mathcal{S}'}\right) = 0$ for $\mathcal{S} \neq \mathcal{S}'$.
\end{thm}
\begin{proof}
	See Appendix \ref{theorem_3}.
\end{proof}
The theorem above suggests that given $\mathbf{B}_\star$, the sampling kernels in the CMS architecture are obtained as
\begin{align}\label{sampler_optimal}
	\psi_p(t) = \sum_{i=1}^{I}\sum_{k\in\mathcal{K}}\sum_{q\in\mathcal{Q}} b_{p,(i,k,q)}^\star \phi_{i,k,q}(t),~~p=1,\cdots,P.
\end{align}
Note that, as long as the preamble sequences do not change, the optimal matrix $\mathbf{B}_\star$ and the corresponding sampling kernels $\psi_p(t)$ are pre-computed only once and their design does not contribute to the running cost of the receiver operations. If the projections on the sampling kernels are implemented in the digital domain instead of being analog filters, then the samples of $\psi_p(t)$ are placed in the static memory  that contains the receiver signal processing algorithms.

If the principal eigenvectors $\mathbf{U}_P$ are unique, then the choice of $\mathbf{B}_\star$ in general spreads out the pairwise KL distances, but $\mathbb{D}\left(\mathcal{H}_{\mathcal{S}}\|\mathcal{H}_{\mathcal{S}'}\right) = 0$ is possible for some choice of $\mathcal{S}$ and $\mathcal{S}'$. If $\mathrm{spark}(\mathbf{M})\geq 2|\mathcal{I}|R$ and we choose $P=\mathrm{rank}[\mathbf{M}]$, then it can be ensured that $\overline{\mathbb{D}}$ is maximized and $\mathbb{D}\left(\mathcal{H}_{\mathcal{S}}\|\mathcal{H}_{\mathcal{S}'}\right) > 0$ is guaranteed for $\mathcal{S} \neq \mathcal{S}'$ with $|\mathcal{S}|,|\mathcal{S}'|\leq |\mathcal{I}|R$. On the other hand, an extreme example where the eigenvectors are not unique is when $\{\phi_{i,k,q}(t)\}_{i=1,\cdots,I}^{k\in\mathcal{K},q\in\mathcal{Q}}$ form an orthogonal basis such that $\mathbf{M}=\mathbf{I}$. In this case, Theorem 3 is analogous to the fundamental criterion in compressed sensing that aims to find a matrix with $\mathrm{spark}(\mathbf{U}_P^H)\geq 2|\mathcal{I}|R$ that guarantees the recovery of any $|\mathcal{I}|R$-sparse vectors.

{ {\it Remark:} The number $|\mathcal{U}|$ of possible eigen-decompositions of the given matrix $\mathbf{M}$ may be quite large. For the extreme case when $\mathbf{M} = \mathbf{I}$, an arbitrary unitary matrix will be a possible choice. Fortunately, it is well known in compressive sensing that partial unitary matrices (such as the partial DFT matrix) have good compressive sensing properties (mutual coherence), thus this would not entail much loss if the matrix does not have exactly the maximum spark. On the other hand, as long as the number $|\mathcal{U}|$ is small, a finite search is also possible. More importantly, this task only needs to be done once and off-line.}
%
%
%
%
\section{Numerical Results}\label{numerical_results}
In this section, we compare the C-SA acquisition scheme using the CMS architecture we propose against the alternative schemes. The first alternative is the D-SA scheme discussed in Section \ref{sampling_acquisition_SA}, which processes the uncompressed Nyquist-rate samples $\mathbf{c}_{\textrm{\tiny DS}}[n]$ using sparse recovery methods. {The comparison with the D-SA scheme can also be viewed as a comparison with the prior art in \cite{daniele2010sparsity,applebaum2011asynchronous}, which propose the SA scheme for CDMA users using direct sampling}\footnote{There are some subtle differences in the model. For example, the online sequential detection and unknown frequency offsets were not considered in \cite{daniele2010sparsity,applebaum2011asynchronous}. Hence, strictly speaking D-SA is a generalized version of \cite{daniele2010sparsity,applebaum2011asynchronous}, where the same basic idea is expanded to handle a larger set of hypotheses.}. Another alternative is the MF receiver, which also processes uncompressed samples. Rather than exploiting the underlying sparsity of the signal, it uses a filterbank matching the signal with all possible templates considered as hypotheses, trying to identify the link parameters through the best (highest) match.

To benchmark our C-SA against the D-SA and MF schemes we simulate the link acquisition of a single receiver plugged in a network populated by $I=10$ users, out of which $|\mathcal I|=4$ are randomly chosen to be actively transmitting. The user signature codes $a_i[m]$'s in \eqref{preamble} belong to a set of M-sequences \cite{buracas2002efficient}, which are quasi-orthogonal BPSK sequences of length $M=255$ with unit power ($|a_i[m]|^2=1$). Due to the user dislocation, mobility and possible scattering, each path of this asynchronous multi-user channel is characterized by the triplet $\{h_{i,r},t_{i,r},\omega_{i,r}\}_{i\in\mathcal{I},r\in\{1,\ldots,R\}}$ with $R=2$, where $\{h_{i,r}\}$ are Rayleigh distributed, $h_{i,r} \thicksim \mathcal{CN}\big(0,1/|\mathcal{I}|R\big)$, and uncorrelated normalized fading coefficients $\mathbb{E}\{h_{i,r}h_{i',r'}^{\ast}\} = \delta[i-i']\delta[r-r']\cdot 1/|\mathcal{I}|R$. The random delays $\{t_{i,r}\}$ are the sum of: (i) a time of arrival $t_0= \min_{i,r} t_{i,r}$ that is uniformly distributed over an interval that spans the duration of the preamble $t_0 \in \mathcal{U}(0,MT)$, and of (ii) multipath delays that are uniformly distributed within an interval $(t_0,t_0+\tilde{\tau}_{\max})$ where $\tilde{\tau}_{\max}$ is the multipath channel delay spread. Consequently, all the arrival times are within a window of duration $t_0+MT+\tilde{\tau}_{\max}$. The random frequency offsets $\{\omega_{i,r}\}$ are uniformly distributed, $\omega_{i,r} \in \mathcal{U}(-\omega_{\max},\omega_{\max})$, over a range delimited at each side by the maximum Doppler spread $\omega_{\max}$. As we simulate underspread channel conditions, we choose $\omega_{\max}$ such that $\omega_{\max}\tilde{\tau}_{\max} \ll 2\pi$. Thus, for a multipath delay spread of $\tilde{\tau}_{\max}=4T$, the choice of $\omega_{\max}=2.5 \cdot 10^{-3}\times 2\pi / T$ is comparable to a $25$ kHz offset for $1$ MHz signals.}


Specifically, we compare the C-SA with the D-SA and the MF at the same resolution,
which for frequency is $\Delta\omega=\omega_{\max}/5=0.5\times 2\pi/T$ and $\Delta \tau=T/2$ and thus $\mathcal{K}=\{-5,\cdots,5\}$. Given a multipath delay spread of $\tilde{\tau}_{\max}=4T$ and a shift of size $D=10T$, the delay space $\mathcal{Q}$ accounts for a composite delay spread of $\tau_{\max}\geq D + \tilde{\tau}_{\max}=14T$ and therefore $\mathcal{Q}=\{0,\cdots,27\}$. This parameter discretization leads to a multi-user time-frequency grid of $I|\mathcal{K}||\mathcal{Q}|=2640$ elements for the C-SA and D-SA scheme.


\begin{figure}[t]
\centering
\begin{center}
\includegraphics[scale=0.5]{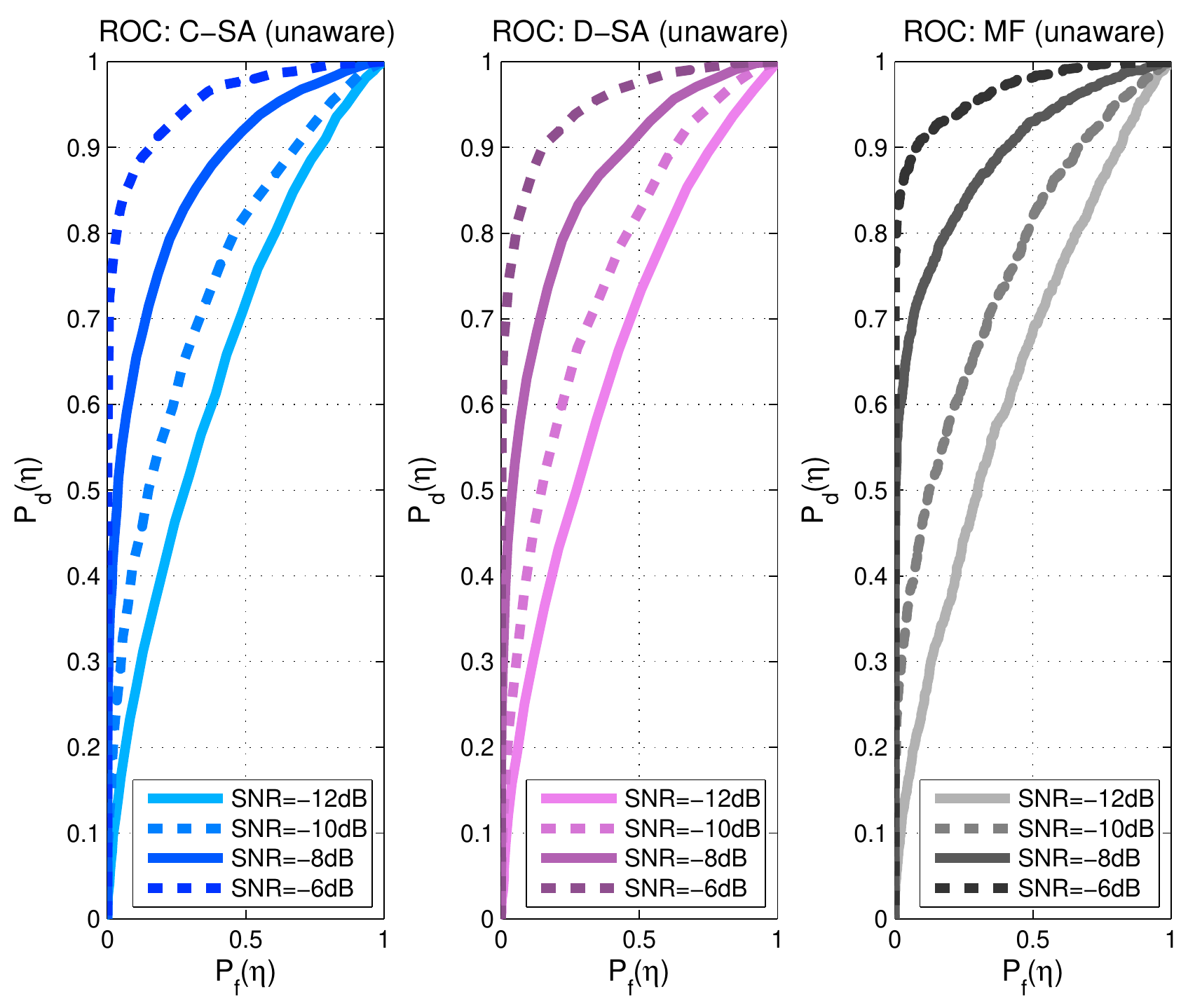}
\caption{{ Comparison of ROC curves for the order unaware receiver using the C-SA scheme with $P=80$ channels (left), the D-SA scheme processing $M$ uncompressed Nyquist observations (middle), and the MF scheme (right); tested at SNR=$\{-12,-10,-8,-6\}$ dB }}\label{fig:roc2}
\end{center}
\vspace{-0.6cm}
\end{figure}

We test the C-SA using three different numbers of sampling channels $P=\{60,80,100\}$. The D-SA scheme uses Nyquist-rate samples per shift, which corresponds to using the whole spreading code duration of $255$ samples, 
while the MF scheme uses a $I|\mathcal{K}|=110$-channel filterbank and performs $I|\mathcal{K}||\mathcal{Q}|=2640$ projections per shift with $|\mathcal{Q}|$ samples on each branch. In the simulations, the C-SA compressive samples are generated from the same Nyquist-rate samples used for the other receivers, by projecting them onto the digitized version of the sampling kernels $\psi_p(mT_s)$, $p=1,\ldots, P$, $m=0,\ldots, M-1$.  The C-SA simulation recovers the link parameters by solving \eqref{sparse_est} with the OMP algorithm \cite{pati1993orthogonal}, which is a popular choice to approximate the solution of a sparse problem \cite{tropp2007signal}. To motivate our selection of OMP, we refer to Section \ref{complexity} for an empirical evaluation of the OMP against two well-known $\ell_1$ minimizers, SpaRSA \cite{wright2009sparse} and $\ell_1$-Homotopy \cite{salman2010dynamic,malioutov2005homotopy}.

\subsection{Signal Detection Performance}
The first test on the detection performance of the C-SA receiver against the MF and the D-SA receiver we show, is for the cases of completely unknown active user sets, as discussed in Section \ref{unknown_user_case}. In Fig. \ref{fig:roc2}, all receivers are unaware of the random set $\mathcal{I}$ of active users.  Specifically, receivers consider as active components those that are found to have a signal strength that is at least 30\% of the strongest components they estimate, i.e. in \eqref{G_i} $ \rho_i=\max_{k,q} |\widehat{\beta}_{i,k,q}|^2 \, /3$. If no possible component meets this requirement, the channel is declared idle. Events for the winning hypothesis $\mathcal{H}_{\widehat{\mathcal{A}}_{\ell^\star}}$ are generated according to \eqref{asyn:eq8}. To first compare the sensitivity of the different receivers to active components, we define a signal hypothesis $\mathcal{H}_{1} $ corresponding to the all the non-idle channel hypotheses, i.e. $\widehat{\mathcal{A}}_{\ell^\star}\neq  \varnothing$. Then, the detection sensitivity is measured in terms of the receiver operating characteristic (ROC) curve, tracing the probability of detection $P_d(\eta_0) =P\big( \eta_{\textrm{\tiny C-SA}}(\ell_\star) \geq \eta_0 | \mathcal{H}_{1} \big)$, against the probability of false alarm $P_f(\eta_0)=P\big( \eta_{\textrm{\tiny C-SA}}(\ell_\star) \geq \eta_0 | \mathcal{H}_{\varnothing}\big)$ when the channel is actually idle. Note that a positive detection may correspond to an incorrect identification of the specific users that are active. Thus, Section \ref{sec.user-detection} shows the rate of correct detection of active components for the same simulation scenario.

As it can be observed, although the C-SA receiver exploits less than $P/M=80/255 \approx 1/3$ of the Nyquist-rate samples, the results from Fig. \ref{fig:roc2} show a modest degradation of the ROC compared to the MF receiver (less than $0.1$ measured at $P_f(\eta_0)=0.1$ and $\mathrm{SNR}=-8$ dB). { As expected, since the D-SA can leverage the additional observations to enhance its sensitivity, a growing gap $P_d(\eta_{\textrm{\tiny D-SA}})-P_d(\eta_{\textrm{\tiny C-SA}})$ is observable as the SNR increases (measured at $P_f(\eta_0)=0.1$) between the $\mathrm{SNR}=-6$ dB and the $\mathrm{SNR}=-12$ dB curves.}

%
%
%

\begin{figure*}[!t]
\begin{center}
{\subfigure[][Partial User Knowledge]{\resizebox{0.4\textwidth}{!}{\includegraphics{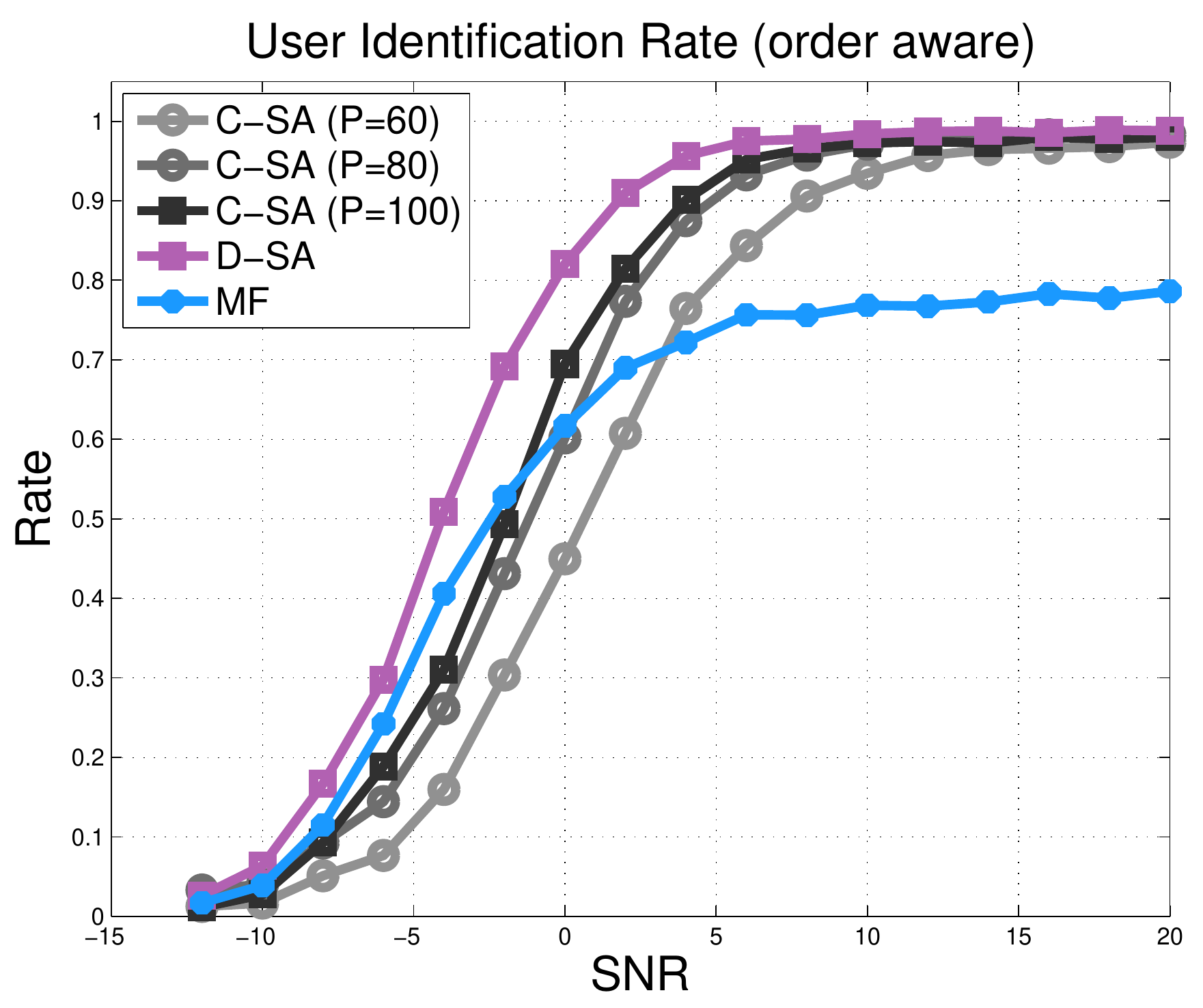}\label{fig:uid1}}}}
 \vspace{0.2cm}
{\subfigure[][No User Knowledge]{\resizebox{0.4\textwidth}{!}{\includegraphics{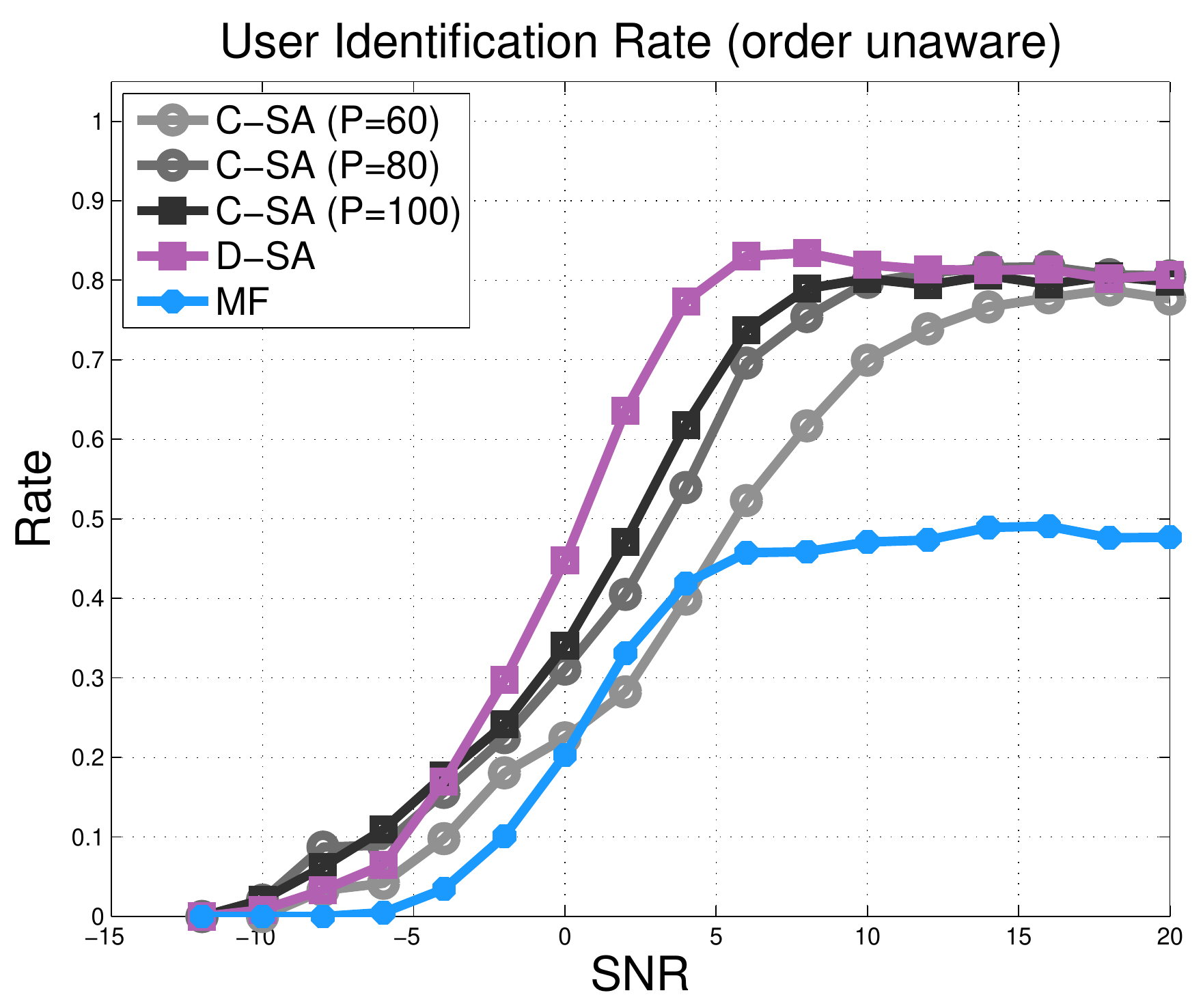}\label{fig:uid2}}}}
\end{center}\vspace{-0.4cm}
\caption{(a) Successful user identification rate $P(\widehat{\mathcal{I}} = \mathcal{I})$ for the order aware receiver implementing the MF (blue) receiver, { the D-SA receiver (red)}, the CMS with C-SA (grey shades) receiver with $P=\{60,80,100\}$. (b) Successful user identification rate $P(\widehat{\mathcal{I}} = \mathcal{I})$ for the order unaware receiver implementing the MF (blue) receiver, { the D-SA receiver (red)}, or the CMS (grey shades) receiver with $P=\{60,80,100\}$}\vspace{-0.5cm}
\end{figure*}

\begin{figure*}[!t]
\begin{center}
{\subfigure[][$\text{RMSE}(\bdsb{\tau})$]{\resizebox{0.4\textwidth}{!}{\includegraphics{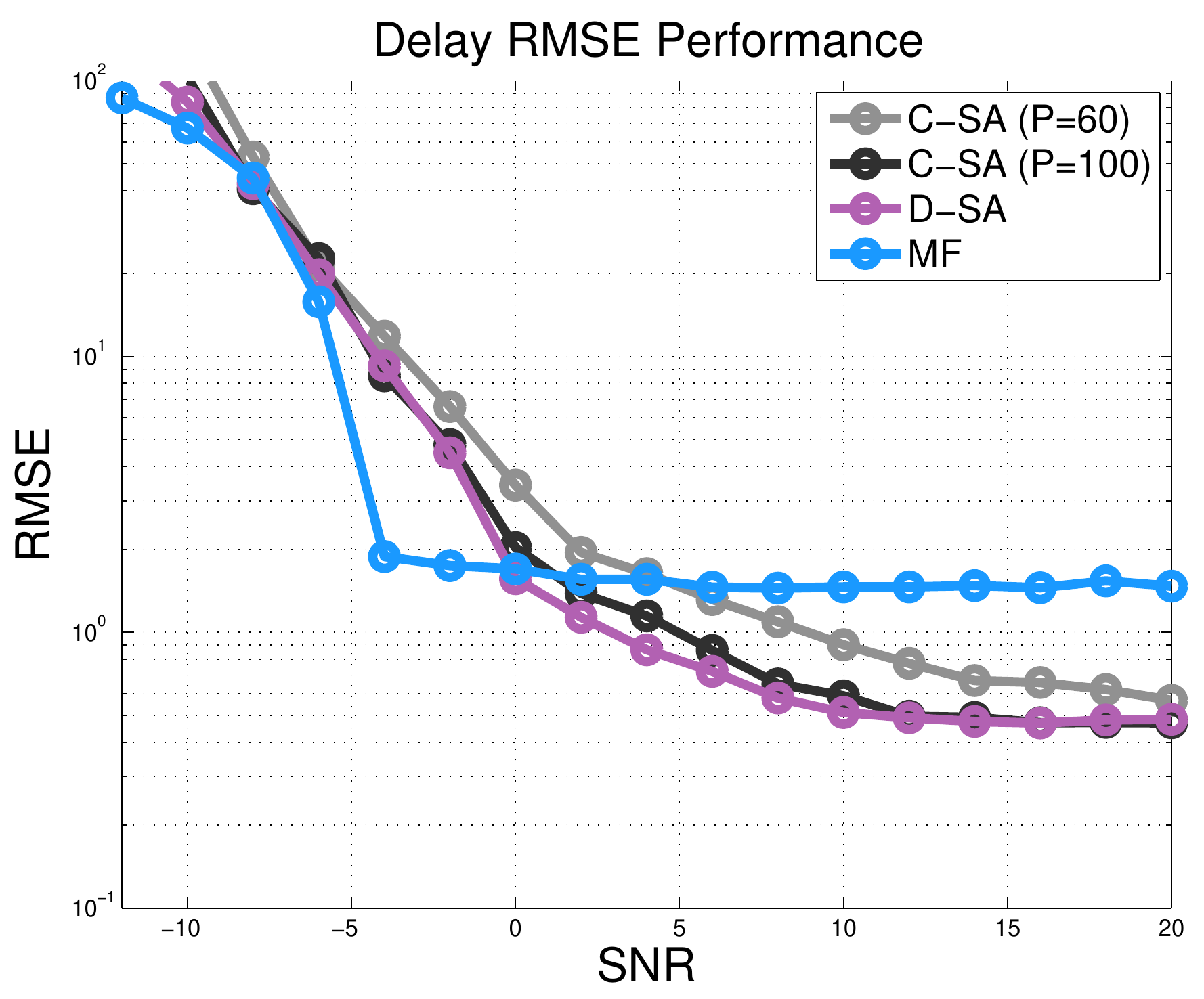}\label{fig:uid3}}}}
 \vspace{0.2cm}
{\subfigure[][$\text{RMSE}(\bdsb{\omega})$]{\resizebox{0.4\textwidth}{!}{\includegraphics{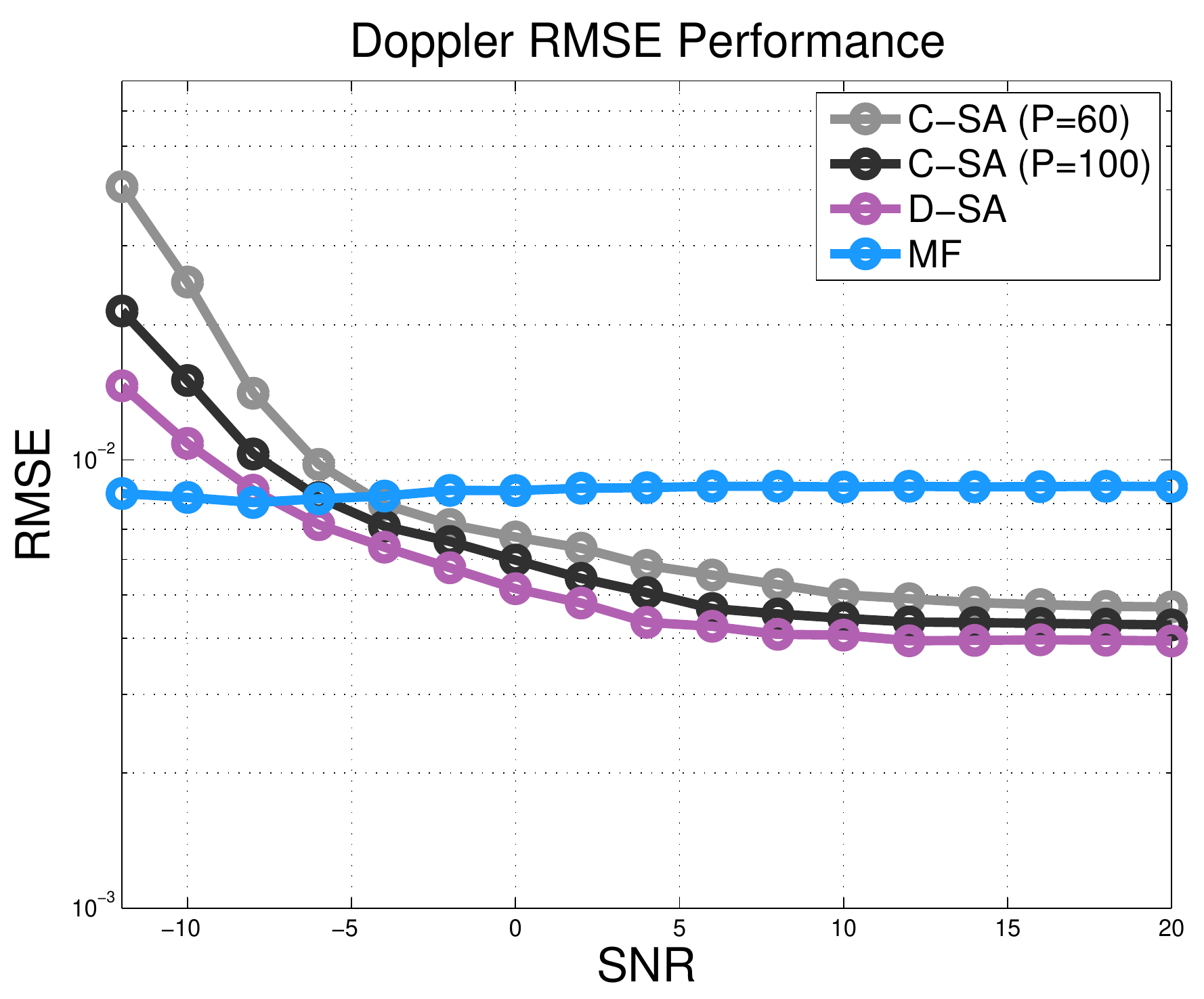}\label{fig:uid4}}}}
\end{center}\vspace{-0.4cm}
\caption{(a) $\text{RMSE}(\bdsb{\tau})$ as a function of $\text{SNR}=\{-12,-8,\ldots,20\}$ implemented by the C-SA scheme (grey shades) with $P=\{60,100\}$, { by the D-SA (red)} or by the MF (blue). (b) $\text{RMSE}(\bdsb{\omega})$ as a function of $\text{SNR}=\{-12,-8,\ldots,20\}$ implemented by the C-SA scheme (grey shades) with $P=\{60,100\}$, { by the D-SA (red)} or by the MF (blue).}\vspace{-0.5cm}
\end{figure*}

\subsection{User Identification and Parameter Estimation Performance}\label{sec.user-detection}

In the simulations shown in Fig. \ref{fig:uid1} and \ref{fig:uid2}, we measure the detection performance of $\widehat{\mathcal{I}}$ by the rate of successful identification $P(\widehat{\mathcal{I}} = \mathcal{I})$. In these simulations, the threshold $\eta_0$ is set to a level such that $P_f(\eta_0)\leq 0.1$, and thus we trace the curve at the point $P_f(\eta_0)=0.1$ on the figures.
The two sets of figures correspond to, respectively, the case where a receiver has partial knowledge of the active user components (as in for example the GPS receiver discussed in Section \ref{partial_user_case}) and the exact same case examined previously in Fig. \ref{fig:roc2}, where the receivers are unaware of the user $\mathcal{I}$. In the second case, for successful detection, not only the elements of the sets have to be consistent $\widehat{\mathcal{I}}\subseteq\mathcal{I}$ but also their cardinality needs to be identical $|\widehat{\mathcal{I}}|=|\mathcal{I}|$. Instead for the first case, if  the receiver has partial knowledge of the active components, then what matters is that the components are correctly identified, but their number is known ahead of time.

With a relatively short training sequence, we can see in Fig. \ref{fig:uid1} that the C-SA in the first case identifies the active user set with large probability ($0.96$ at $\mathrm{SNR}=20$ dB). The MF has  worse performance due to the multi-user interference and to the presence of unresolvable paths. In fact, the MF receiver is unable to isolate the multi-path arrivals that fall within the same symbol period and, due to the presence of different Dopplers, its side-lobes may contribute negatively to the correlation, masking other active components. In contrast, the OMP algorithm in the C-SA scheme cancels the contributions from paths detected in previous iterations, before updating the projections to search for other components (the OMP processing steps are summarized in Section \ref{complexity}). It is evident, however, that a low SNR, the C-SA scheme suffers from a loss due to the compression ($-1$ dB at the rate $0.6$ with $P=100$). { This is clearly understood by observing the performance of the D-SA receiver as well. By processing uncompressed samples with the sparse reconstruction method, the D-SA curve combines the best of both worlds and, thus, its performance bounds the user identification rate for both the MF and the C-SA receivers in both examples.} As shown in Fig. \ref{fig:uid2}, the performance degrades when the receivers do not have side information on $|\mathcal{I}|$ (a difference of $-0.13$ for the CMS receiver against $-0.3$ of the MF at $\mathrm{SNR}=20$ dB). This is due to the cardinality mismatch, $\{|\widehat{\mathcal{I}}| \neq |\mathcal{I}|\}$, that occurs while estimating the order.
%
%
%

\begin{figure*}[!t]
\begin{center}
{\subfigure[][ROC and Identification Rate]{\resizebox{0.45\textwidth}{!}{\includegraphics{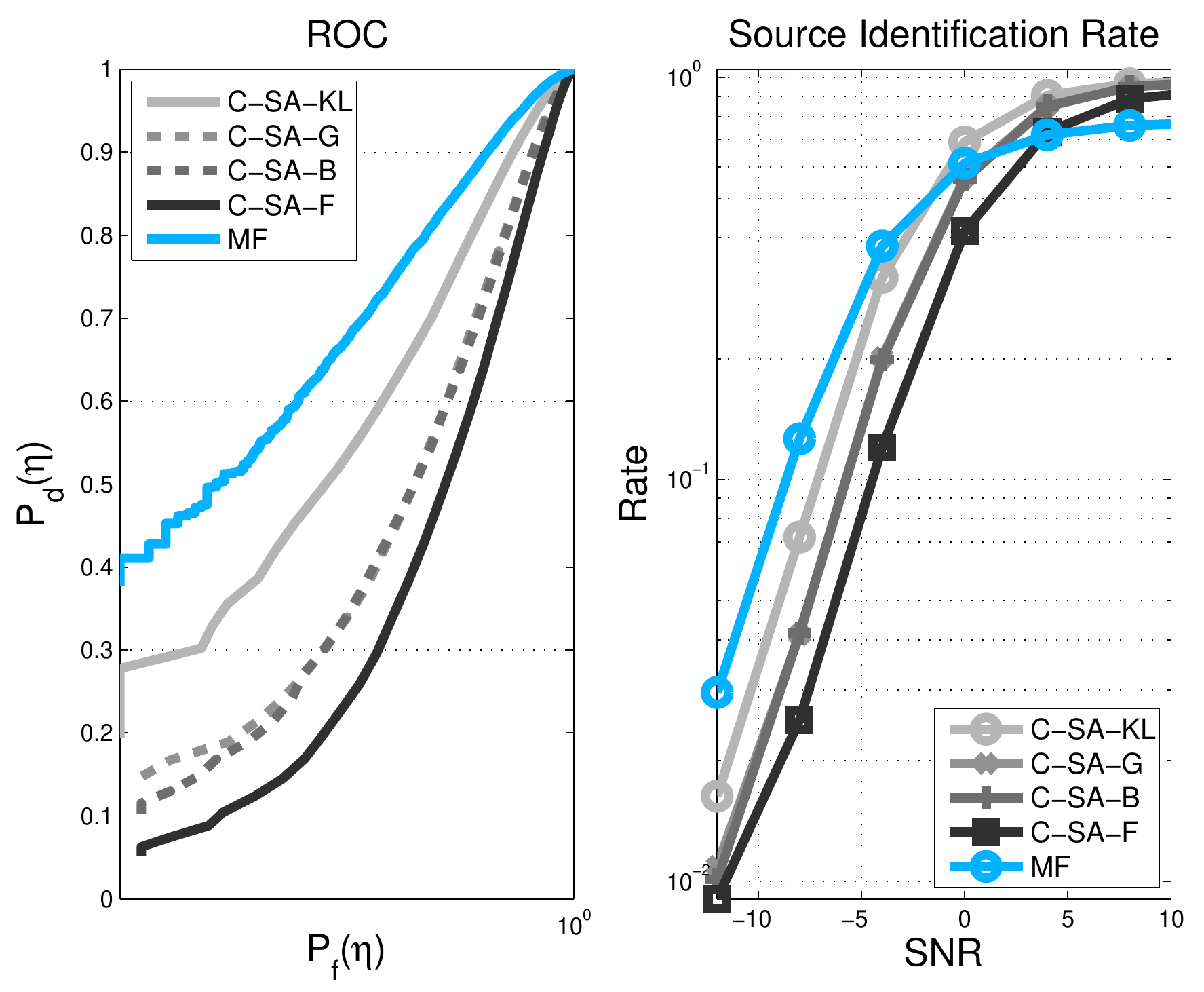}\label{fig:plot1}}}}
 \vspace{0.2cm}
{\subfigure[][Parameter RMSE]{\resizebox{0.45\textwidth}{!}{\includegraphics{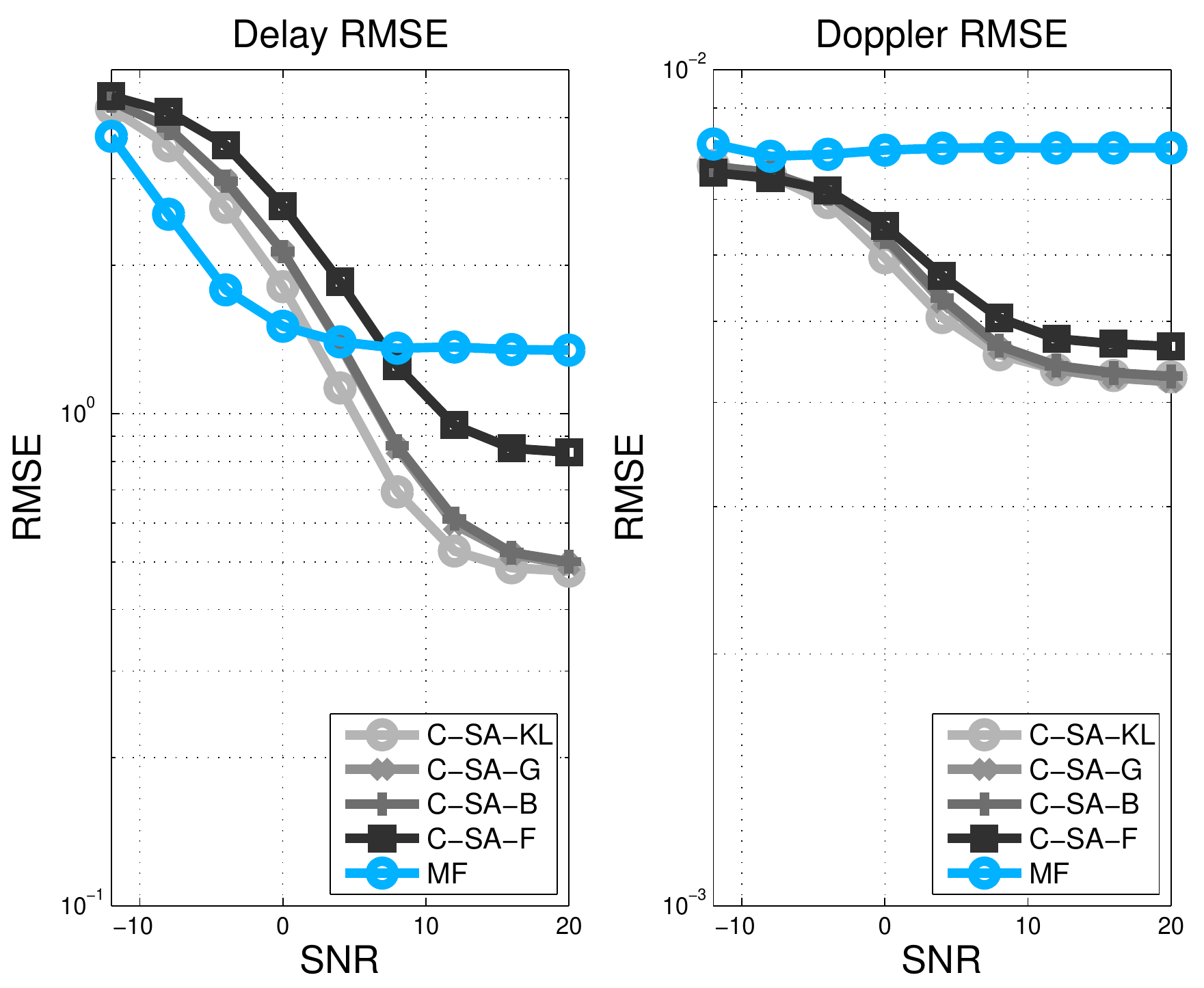}\label{fig:plot2}}}}
\end{center}\vspace{-0.4cm}
\caption{(a) ROC curve at SNR=$-8$ dB and user identification rate $P(\widehat{\mathcal{I}} = \mathcal{I})$ of the C-SA scheme, against the MF receiver and different choices of $\mathbf{B}$'s. (b) Delay and Doppler estimation RMSE of the C-SA scheme, against different random designs of $\mathbf{B}$ and the MF.}\vspace{-0.5cm}
\end{figure*}

The accuracy of the recovered set $\widehat{\mathcal{A}}_{\ell^\star}$ is evaluated by the root mean square error (RMSE) of $\tau_{i,r}$ and $\omega_{i,r}$ that are associated with the correctly identified users $\mathcal{I} = \widehat{\mathcal{I}}$. Thus,
\begin{align}
\mathrm{RMSE}(\bdsb{\tau}) &\triangleq \sqrt{\frac{1}{|\widehat{\mathcal{I}} \cap \mathcal{I}|}\frac{1}{R}\sum_{r=1}^R\, \sum_{i \in \{\widehat{\mathcal{I}} \cap \mathcal{I}\}}\big(\tau_{i,r} - \widehat{q}_{i,r}\Delta \tau \big)^2 }\nonumber\\
\mathrm{RMSE}(\bdsb{\omega}) &\triangleq \sqrt{\frac{1}{|\widehat{\mathcal{I}} \cap \mathcal{I}|}\frac{1}{R}\sum_{r=1}^R \sum_{i \in \{\widehat{\mathcal{I}} \cap \mathcal{I}\}}  \big(\omega_{i,r} - \widehat{k}_{i,r}\Delta \omega \big)^2 }\nonumber
\end{align}
are the average RMSE of the delay parameters the Doppler frequencies respectively.

To verify the accuracy of the parameter estimates of $\tau_{i,r}$ and $\omega_{i,r}$, we trace the RMSE's of the order-aware case. Once again we observe, from Fig. \ref{fig:uid3} and Fig. \ref{fig:uid4}, that the performance of RMSE($\bdsb{\tau}$) and RMSE($\bdsb{\omega}$) is enhanced by the detector that better leverages the presence of the multi-path. In fact, at $\text{SNR}=20$ dB, the accuracy of the CMS, with $P=100$, { and the D-SA} approach the grid resolution, i.e. $\text{RMSE}(\bdsb{\tau}) \approx \dtau$ and $\text{RMSE}(\bdsb{\omega}) \approx \dw $. Oppositely, the contribution of the unresolvable paths to the correlation, in either frequency or time, adversely affects the parameter selection. Not canceling the previously selected components contributes to a large error after the selection of the dominant paths as the same arrivals are likely to be selected more than once by the presence of correlated components. These errors impact the highest resolution since at $\text{SNR}=20$ dB: $\text{RMSE}(\bdsb{\tau})> 2\dtau$ and $\text{RMSE}(\bdsb{\omega}) > 2\dw $. Instead, at low SNR, the performance is bounded by the maximum error given by the search space which is a function of $\omega_{\max}$ and $\tau_{\max}$, respectively, due to the early detection resulting from heavy noise.

%
%

\subsection{Optimality of Compressive Samplers}

In this subsection, we briefly compare the performances of the C-SA scheme using a $P=100$-channel CMS architecture with the optimal samplers versus other random projection schemes in compressed sensing. The ROC curve and the user identification rate $P(\widehat{\mathcal{I}}=\mathcal{I)}$ in Fig. \ref{fig:plot1} show that the optimal sampling kernel, denoted by C-SA-KL, exhibits  better performance than random designs of $\mathbf{B}$ using matrices whose entries are Gaussian (C-SA-G), Bernoulli (C-SA-B), or randomly selected rows of a DFT matrix (C-SA-F). It can also be observed from Fig. \ref{fig:plot2} that the RMSE of the delay and Doppler estimates are also improved.

{
\section{Cost Analysis} \label{complexity}
In this section, we explicitly analyze the implementation costs of the MF and the proposed C-SA architectures in terms of storage requirement and computational complexity. The cost benefits are illustrated in two regimes respectively: the {\it analog implementation}, that corresponds to what the paper describes mathematically in detail, and a {\it digital implementation}, which would be necessary if the projections on the compressive samplers $\psi_p(t)$ cannot be implemented in the form of analog filters. The compressive procedure in that case would emulate our implementation in simulations, where Nyquist samples of the signal $x(t)$ are projected onto digital filters matched to the samples of $\psi_p(t)$.

The metric to evaluate {\it storage requirements} is given by the A/D hardware cost, measured as the number of sampling channels $P$ which is the buffer size of the A/D samples, while the {\it computational complexity} is evaluated counting the {\it total amount of complex additions and multiplications}, and by the {average CPU run time} spent on executing tasks (a 64-bit i7 920 CPU running at 2.67 GHz). In the following, we first settle on the sparse recovery solver for the SR-LRT in the C-SA scheme, and then proceed with our comparison using the chosen solver.

\begin{figure}[!t]
\begin{center}
{\subfigure[][run time v.s. the number of sampling channels $P$]{\resizebox{0.4\textwidth}{!}{\includegraphics{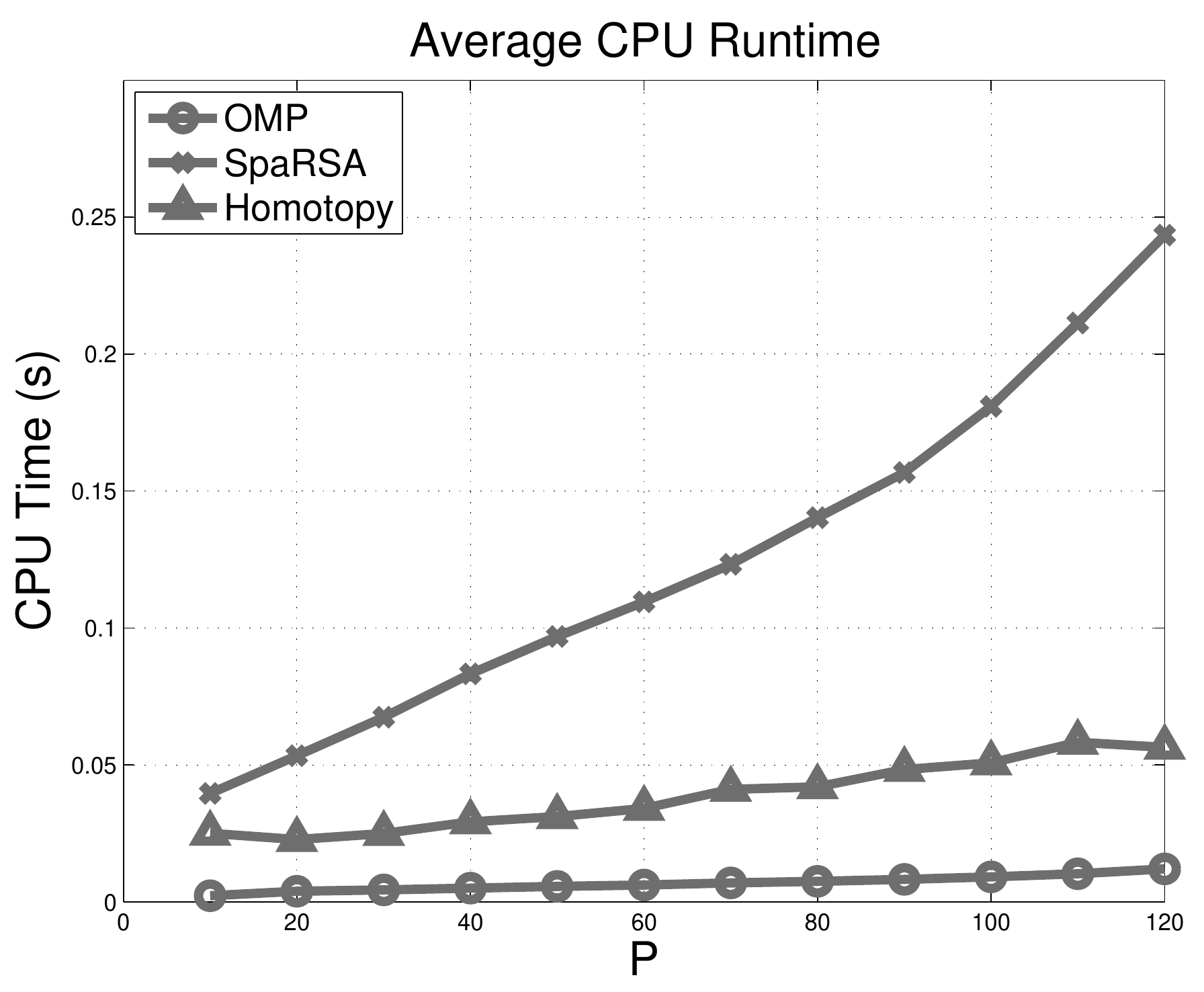}\label{fig:alg1}}}}
 \vspace{0.2cm}
{\subfigure[][run time v.s. the length of the preamble sequence $M$ ]{\resizebox{0.4\textwidth}{!}{\includegraphics{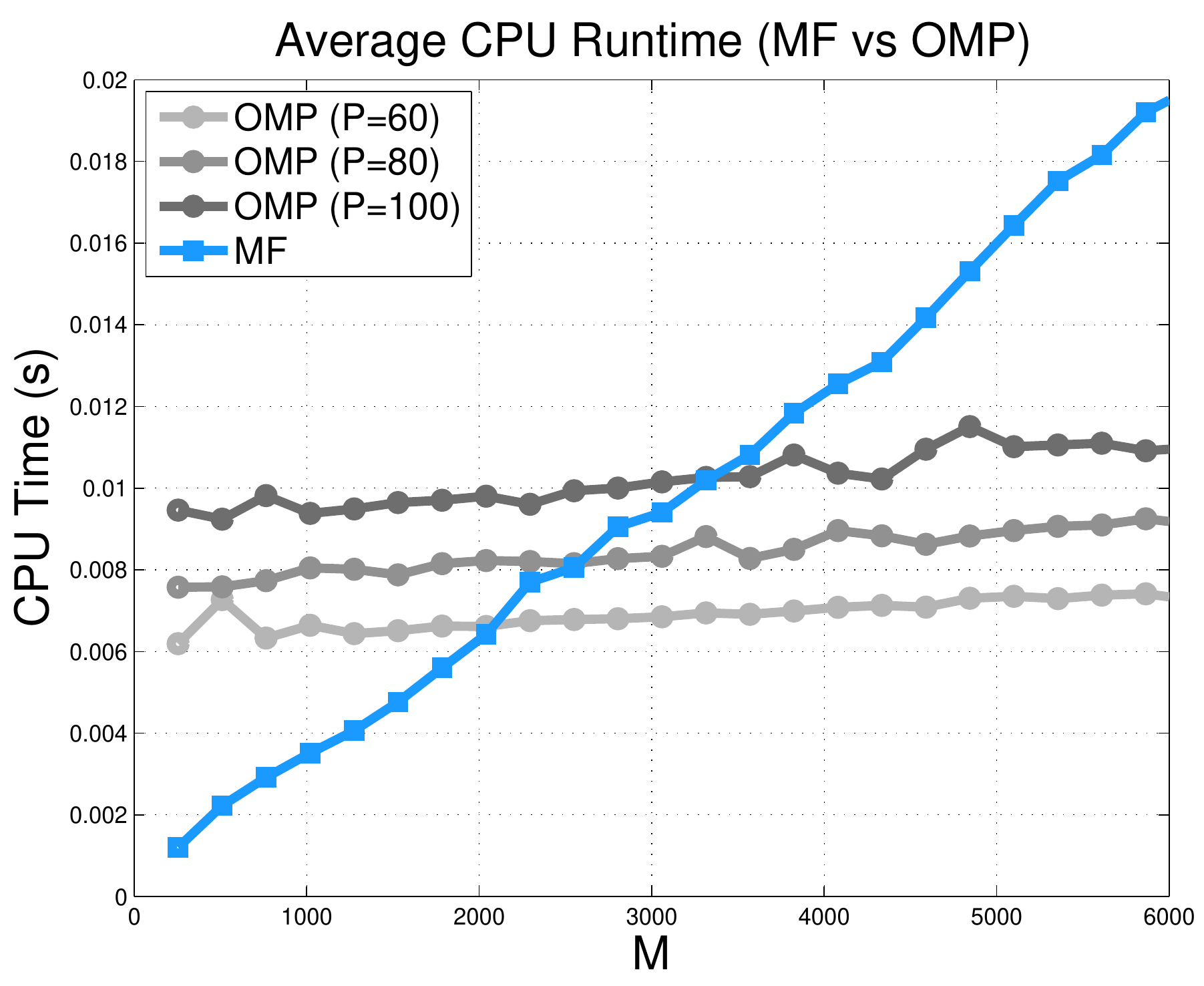}\label{fig:alg2}}}}
\end{center}
\vspace{-0.4cm}
\caption{(a) Average CPU run time for the CMS receiver, implemented with different solvers (OMP, SpaRSA, $\ell_1$-Homotopy), as a function of $P=\{ 10,20,\cdots,120\}$. (b) Average CPU runtime of a CMS receiver using $P=\{60,80,100\}$ compressed observations against the MF receiver. The average runtime is measured against the preamble length $M$.}
\end{figure}

\subsection{Sparse Recovery Solver: The OMP Algorithm}
The C-SA receiver spends its greatest effort in solving the optimization \eqref{sparse_est}. Fast $\ell_1$ minimizers like SpaRSA \cite{wright2009sparse} or $\ell_1$-Homotopy \cite{salman2010dynamic,malioutov2005homotopy} are often the methods of choice. The former greatly reduces the complexity by approximating the Hessian of the gradient descent by a diagonal matrix, whereas the latter inverts a system of equations whose number of unknowns, at each iteration, remains restricted to the non-zero elements of the sparse vector estimate. Greedy algorithms like the OMP \cite{pati1993orthogonal}, are efficient approximations to the solution of sparse problems as well \cite{tropp2007signal}. The OMP algorithm iteratively detects the strongest element in the sparse vector and removes its contribution in the next iteration; thus, the number of iterations required by OMP is bounded by the maximum possible components that exist in the signal, which in our case is $|\mathcal{I}|R$. The average CPU run time spent in solving \eqref{sparse_est} with different solvers is illustrated in Fig. \ref{fig:alg1} traced against $P$, where the OMP algorithm shows significantly less computation time. Thus, in Fig. \ref{fig:alg2} we further compare the average CPU run time of the CMS receiver using OMP against the MF receiver, the implementation details of which will be discussed in the following subsection. OMP has smaller complexity primarily because it stops as soon as all the strong entries have been detected. In contrast, $\ell_1$-Homotopy and SpaRSA do not limit the search to a single set, but rather explore the feasible set by selecting and de-selecting elements of the support vector ($\ell_1$-Homotopy), or by shrinking it through a gradient descent (SpaRSA), until a desired convergence criterion has been met.


\subsection{C-SA Scheme v.s. MF Scheme}
\begin{algorithm}[htp]
\caption{C-SA Scheme}\label{CMS_algorithm}
\begin{itemize}
    \setlength{\itemindent}{2em}
    \item[({\footnotesize\footnotesize\bf CMS.1})] obtain compressive samples $\mathbf{c}[n]$ at the $n$th shift;
	\item[({\footnotesize\bf CMS.2})] initialize $\bdsb{\beta}^0=\mathbf{0}$, $\mathcal{A}_n^0=\varnothing$, $\mathbf{\bar{S}}_0 = \mathbf{B}\mathbf{M}$, $\mathbf{S}_0 = \varnothing$, $j=1$ and run the OMP algorithm;
		 \begin{itemize}
                \item[({\footnotesize  OMP.1})] remove interference $\bdsb{\delta}^j =\mathbf{c}[n]-\mathbf{S}_{j-1}\Big[\bdsb{\beta}^{j-1}\Big]_{\mathcal{A}_n^{j-1}}$;
			    \item[({\footnotesize OMP.2})] projection $\bdsb{\xi}^j=\mathbf{\bar{S}}_{j-1}^T\bdsb{\delta}^j$, $\bdsb{\xi}^j=[\cdots,\xi_{i,k,q}^j,\cdots]^T$;
			    \item[({\footnotesize OMP.3})] detection $\mathcal{A}_n^j = \mathcal{A}_n^{j-1}\bigcup~\{(i,k,q)\}$ with $$(i,k,q) = \arg\max_{i,k,q}~|\xi_{i,k,q}^j|^2;$$
			    \item[({\footnotesize  OMP.4})] update $\mathbf{S}_j=\Big[\mathbf{B}\mathbf{M}\Big]_{(:,\mathcal{A}_n^j)}$ and $\mathbf{\bar{S}}_j=\Big[\mathbf{B}\mathbf{M}\Big]_{(:,\overline{\mathcal{A}_n^j})}$;
                \item[({\footnotesize  OMP.5})] update the link vector
                     \begin{align}
                        \Big[\bdsb{\beta}^j\Big]_{\mathcal{A}_n^j} &= \left(\mathbf{S}_j^T\mathbf{S}_j\right)^{-1}\mathbf{S}_j^T\mathbf{c}[n]\\
                        \Big[\bdsb{\beta}^j\Big]_{\overline{\mathcal{A}_n^j}} &= \mathbf{0};
                     \end{align}
			    \item[({\footnotesize OMP.6})]  stop if either $j=|\mathcal{I}|R$ or $\left\|\mathbf{c}[n]-\mathbf{B}\mathbf{M}\bdsb{\beta}^j\right\|<\epsilon$, and set $j=j+1$.
		\end{itemize}
	\item[({\footnotesize\bf CMS.3})] Evaluate the likelihood ratio $\eta_{\textrm{\tiny C-SA}}(n)$ in \eqref{eta_n} and check if it exceeds $\eta_0$.
       \item[({\footnotesize\bf CMS.4})] If yes, then extract components accordingly (order-aware, order-unaware).
\end{itemize}
\end{algorithm}
%
%
\begin{algorithm}[htp]
\caption{MF Scheme}\label{MF_algorithm}
\begin{itemize}
    \setlength{\itemindent}{2em}
    \item[(\footnotesize \bf MF.1)] obtain the sample array $\mathbf{C}_{\textrm{\tiny MF}}[n]$ in \eqref{sample_array} from the MF filterbank;

    \item[(\footnotesize \bf MF.2)] identify the maximum output and check if it exceeds $\rho_i$ for all $i=1,\cdots,I$;
    \item[(\footnotesize \bf MF.3)] If yes, then extract the delay-Doppler set for each active user as in Section \ref{sampling_acquisition}.
\end{itemize}
\end{algorithm}

Using the OMP algorithm for sparse recovery, we summarize the steps of the C-SA and MF schemes in Algorithms \ref{CMS_algorithm} and \ref{MF_algorithm} respectively. Based on the algorithm descriptions, we provided the order of storage cost and computational complexities in the tables below. Storage accounts for a {\it data path storage} component, dynamically updated with the streaming data that correspond to new observations to be processed, and for a {\it static component}, that stores filters or sampling kernels parameters needed to perform signal processing on the data.
\begin{figure*}[htp]
\centering
\begin{tabular}{l*{6}{l}r}
\toprule
\toprule
\cmidrule(r){1-1}
\textbf{CMS Receiver}     &  \textbf{Storage} (analog) & \textbf{Complexity} (analog)  & \textbf{Storage} (digital) & \textbf{Complexity} (digital) \\
\midrule
($\bf{CMS.1}$)      & $\mathcal{O}(P)$  &  $0$   & $\mathcal{O}(MP)$ & $\mathcal{O}(MP)$  \\
($\bf{CMS.2}$)      & $0$  &   $\mathcal{O}(I|\mathcal{K}||\mathcal{Q}|)+\mathcal{O}(P^3)$  & $0$  &   $\mathcal{O}(I|\mathcal{K}||\mathcal{Q}|)+\mathcal{O}(P^3)$\\
($\bf{CMS.3}$)      & $\mathcal{O}(P) $     &  $\mathcal{O}(P)$ & $\mathcal{O}(P) $     &  $\mathcal{O}(P)$ \\
Total      & $\mathcal{O}(P) $     &  $\mathcal{O}(I|\mathcal{K}||\mathcal{Q}|)+\mathcal{O}(P^3)$     &  $\mathcal{O}(MP)$ & $\mathcal{O}(I|\mathcal{K}||\mathcal{Q}|)+\mathcal{O}(P^3)$ \\
\bottomrule
\bottomrule
\end{tabular}
\caption{Complexity breakdown for the CMS receivers at every shift $t=nD$.}\label{tab1:cs_comp}
\end{figure*}
\begin{figure*}[htp]
\centering
\begin{tabular}{l*{6}{l}r}
\toprule
\toprule
\cmidrule(r){1-1}
\textbf{MF Receiver}     &  \textbf{Storage} (analog) & \textbf{Complexity} (analog)  & \textbf{Storage} (digital) & \textbf{Complexity} (digital) \\
\midrule
($\bf{MF.1}$)      & $\mathcal{O}(I|\mathcal{K}||\mathcal{Q}|)$  &  $0$   & $\mathcal{O}(MI|\mathcal{K}||\mathcal{Q}|)$ & $\mathcal{O}(MI|\mathcal{K}||\mathcal{Q}|)$  \\
($\bf{MF.2}$)      & $0$  &   $\mathcal{O}(I|\mathcal{K}||\mathcal{Q}|)$  & $0$  &   $\mathcal{O}(I|\mathcal{K}||\mathcal{Q}|)$ \\
Total      & $\mathcal{O}(I|\mathcal{K}||\mathcal{Q}|) $     &  $\mathcal{O}(I|\mathcal{K}||\mathcal{Q}|)$     &  $\mathcal{O}(MI|\mathcal{K}||\mathcal{Q}|)$ & $\mathcal{O}(MI|\mathcal{K}||\mathcal{Q}|)$ \\
\bottomrule
\bottomrule
\end{tabular}
\caption{Complexity breakdown for the MF receivers at every shift $t=nD$.}\label{tab2:mf_compx}
\end{figure*}
}

{
It is seen in Fig. \ref{tab1:cs_comp} and \ref{tab2:mf_compx} that both our  C-SA scheme with the CMS receiver and the MF receiver (using exhaustive matched filtering) have computational complexities that scale linearly with the dimension of the search space $I|\mathcal{K}||\mathcal{Q}|$. However, the  {\it data path storage} of the C-SA receivers is greatly reduced. Another storage gain is found in the case of digital implementation, because there are fewer projections to be made and thus, in principle, unless the MF are synthesized on the fly, a smaller amount of static memory is required to store the samples of $\psi_p(t)$.

When the architectures are implemented in the digital domain, the C-SA receiver also leads to a great reduction in computational complexity, with an approximate ratio with the MF receiver complexity of
\begin{align}
	\frac{MP + I|\mathcal{K}||\mathcal{Q}| + P^3}{MI|\mathcal{K}||\mathcal{Q}|} \approx \frac{P}{I|\mathcal{K}||\mathcal{Q}|} +  \frac{P^3}{MI|\mathcal{K}||\mathcal{Q}|}.
\end{align}
Clearly, when the preamble sequence is long and the search space is large $MI|\mathcal{K}||\mathcal{Q}| \gg P^3$, this ratio becomes less than $1$ and the C-SA architecture leads to computational savings while, as seen in our simulations, maintaining comparable performance. This is also why the C-SA receiver implemented in the simulation (see Fig. \ref{fig:alg2}) considerably outperforms the MF receiver in terms of average CPU run time for large $M\geq 3000$ with $P=60,80,100$. On the other hand, when $M$ is small (e.g. $M=255$ in Section \ref{numerical_results}), the MF receiver has less computation time for $M<3000$ against $P=60,80,100$ as in Fig. \ref{fig:alg2}, but such a short preamble does not provide sufficient processing gain for reliable link acquisition in the presence of multipath, as can be clearly seen from the numerical results (e.g., see Fig. \ref{fig:uid1}). Thus, when $M$ is small, the gain of the C-SA receiver also lies in the superior acquisition performance demonstrated by the numerical results, except for the low SNR region where the C-SA is not sufficiently sensitive.
}

\section{Conclusions}
In this paper, we proposed the SR-LRT receiver using a unified CMS architecture for link acquisition, which we refer to as the C-SA scheme. This scheme uses a sequential SR-LRT that jointly detects signal presence and recovers the active users with their link parameters. We optimized the CMS architecture to maximize the average Kullback-Leibler distance among the hypotheses tested in the SR-LRT and show that, with the optimal compressive samplers we propose, the receiver detection outperforms those with conventional compressed sensing alternatives. Furthermore, through the numerical comparison of the proposed architecture with the D-SA scheme and the MF approach, we have shown that the C-SA receiver can scale down its processing storage and complexity with greater flexibility, while maintaining satisfactory performance.

\appendices

\section{Proof of Theorem \ref{observation_model}}\label{theorem_1}
Substituting \eqref{sparse_sig} into \eqref{gen_sampling}, we have
\begin{align}\label{gen_sample}
	c_p[n] &= \sum_{i=1}^{I} \sum_{k\in\mathcal{K}}\sum_{q\in\mathcal{Q}}\alpha_{i,k,q}e^{\mathrm{i}k\Delta\omega \ell D}\ip{\phi_{i,k,q}(t-\ell D)}{\psi_p(t-nD)}\nonumber\\
	&+\ip{v(t)}{\psi_p(t-nD)}.
\end{align}
Define the $P\times I|\mathcal{K}||\mathcal{Q}|$ matrix
	\begin{align}\label{lem_1_app}
		 \Big[\mathbf{M}_{\psi\phi}[n-\ell]\Big]_{p, (i,k,q)}
		&= R_{\psi_p\phi_{i,k,q}}\left[(n-\ell)D\right]\\
		&\triangleq \ip{\phi_{i,k,q}(t-\ell D)}{\psi_p(t-nD)}
	\end{align}
and denote $v_p[n]\triangleq\ip{v(t)}{\psi_p(t-nD)}$ as the sample of filtered noise, whose covariance can be obtained as $$\mathbb{E}\{v_p[n]v_{p'}^\ast[n]\} = \sigma^2 \ip{\psi_p(t)}{\psi_{p'}(t)}$$ using $\mathbb{E}\{v(t)v^\ast(s)\}=\sigma^2\delta(t-s)$. Therefore, the noise vector $\bdsb{\nu}[n]\triangleq [\nu_1[n],\cdots,\nu_P[n]]^T$ has a covariance matrix obtained as $\mathbf{R}_{vv} = \mathbb{E}\{\bdsb{\nu}[n]\bdsb{\nu}^H[n]\}=\sigma^2\mathbf{R}_{\psi\psi}$ where $$[\mathbf{R}_{\psi\psi}]_{p,p'}\triangleq\ip{\psi_p(t)}{\psi_{p'}(t)}$$ is the Gram matrix of the kernels $\psi_p(t)$'s.

Denote $\mathbf{c}[n]\triangleq \left[c_1[n],\cdots,c_P[n]\right]^T$ as the length-$P$ vector containing the samples acquired from the CMS filterbank at time $t=nD$. Given the link vector ${ \bdsb{\alpha}[\ell]}=[\cdots,\alpha_{i,k,q},\cdots]^T$ at the $\ell$th shift as \eqref{index_mu}, we then have the observation model in matrix form
\begin{align}
	\mathbf{c}[n] = \mathbf{M}_{\psi\phi}[n-\ell]\bdsb{\Gamma}[\ell]{ \bdsb{\alpha}[\ell]} +\bdsb{\nu}[n].
\end{align}

Using a sampling kernel constructed as
 \begin{align}
	 \psi_p(t) = \sum_{i=1}^{I}\sum_{k\in\mathcal{K}}\sum_{q\in\mathcal{Q}} b_{p,(i,k,q)} \phi_{i,k,q}(t),
\end{align}
the cross-correlation $R_{\psi_p\phi_{i,k,q}}[(n-\ell)D]$ in \eqref{lem_1_app} can be evaluated as
\begin{align}\label{R_psiphi}
	& R_{\psi_p\phi_{i,k,q}}[(n-\ell)D]\\
	&= \sum_{i'=1}^{I}\sum_{q'\in\mathcal{Q}}\sum_{k'\in\mathcal{K}} b_{p,(i',k',q')} R_{\phi_{i',k',q'}\phi_{i,k,q}}[(n-\ell)D],\nonumber
\end{align}
where 
\begin{align*}
	R_{\phi_{i',k',q'}\phi_{i,k,q}}[(n-\ell)D]=\ip{\phi_{i,k,q}(t-\ell D)}{\phi_{i',k',q'}(t-nD)}. 
\end{align*}
With a change of variable $t'=t-nD-q'\Delta\tau$, the correlation can be expressed as
\begin{align}\label{R_phiphi}
	R_{\phi_{i',k',q'}\phi_{i,k,q}}[(n-\ell)D] & = e^{\mathrm{i}k\Delta\omega(n-\ell)D}e^{-jk\Delta\omega q'\Delta\tau}\\
	&\times ~R_{\phi_{i'},\phi_i}^{(k-k')}\left[(q'-q)\Delta\tau + (n-\ell)D\right],\nonumber
\end{align}
where $R_{\phi_{i'},\phi_i}^{(k-k')}(\Delta t)$ is the ambiguity function
\begin{align}
	R_{\phi_{i'},\phi_i}^{(k-k')}(\Delta t) = \int \phi_{i'}^\ast(t)\phi_i(t-\Delta t)e^{-\mathrm{i}(k-k')\Delta\omega t}\mathrm{d}t.
\end{align}
From \eqref{R_psiphi}, $\mathbf{M}_{\psi\phi}[n-\ell]$ in \eqref{lem_1_app} can be re-written as 
\begin{align}
	\mathbf{M}_{\psi\phi}[n-\ell] = \mathbf{B}\mathbf{M}_{\phi\phi}[n-\ell], 
\end{align}	
where
\begin{align*}
	 \Big[\mathbf{M}_{\phi\phi}[n-\ell]\Big]_{(i',k',q'),(i,k,q)}=R_{\phi_{i',k',q'}\phi_{i,k,q}}[(n-\ell)D].
\end{align*}
Then the observation model can be re-written as
\begin{align}
	\mathbf{c}[n] = \mathbf{B} \mathbf{M}_{\phi\phi}[n-\ell]\bdsb{\Gamma}[\ell]{ \bdsb{\alpha}[\ell]} +\bdsb{\nu}[n].
\end{align}
Finally, the Gram matrix of $\psi_p(t)$'s is obtained accordingly as
\begin{align}
	\mathbf{R}_{\psi\psi}=\mathbf{B}\mathbf{M}_{\phi\phi}[0]\mathbf{B}^H,
\end{align}
which gives the noise covariance as $\mathbf{R}_{vv} = \sigma^2\mathbf{B}\mathbf{M}_{\phi\phi}[0]\mathbf{B}^H$.

\section{Proof of Theorem \ref{observation_model_reform}}\label{theorem_2}
From \eqref{gen_sample}, each sample $c_p[n]$ from the CMS architecture $p=1,\cdots,P$ can be expressed as
\begin{align}\label{sample}
	c_p[n] &=
	 \sum_{i'=1}^{I}\sum_{k'\in\mathcal{K}}\sum_{q'\in\mathcal{Q}} b_{p,(i',k',q')} \sum_{i=1}^{I}\sum_{k\in\mathcal{K}}\sum_{q\in\mathcal{Q}}\alpha_{i,k,q}e^{\mathrm{i}k\Delta\omega\ell D}\nonumber\\
	&~\times R_{\phi_{i',k',q'}\phi_{i,k,q}}[(n-\ell)D] + v_p[n].
\end{align}
The summation $\sum_{q\in\mathcal{Q}}  \alpha_{i,k,q}e^{\mathrm{i}k\Delta\omega\ell D}R_{\phi_{i',k',q'}\phi_{i,k,q}}[(n-\ell)D]$ can be adjusted with respect to the relative time index $[n-\ell]$ by re-writing the correlation in \eqref{R_phiphi} as
\begin{align}\label{R_phiphi_exp}
	 &R_{\phi_{i',k',q'}\phi_{i,k,q}}[(n-\ell)D]\\
	 &=e^{\mathrm{i}k\Delta\omega(n-\ell)D}e^{-jk\Delta\omega q'\Delta\tau}
	 R_{\phi_{i'},\phi_i}^{(k-k')}\left[(q'-q)\Delta\tau + (n-\ell)D\right]\nonumber\\
	 &=e^{\mathrm{i}k\Delta\omega(n-\ell)D}e^{-jk\Delta\omega q'\Delta\tau}
	 R_{\phi_{i'},\phi_i}^{(k-k')}\left[(q'-[q+(\ell-n)N])\Delta\tau\right]\nonumber\\
        &=e^{\mathrm{i}k\Delta\omega(n-\ell)D} R_{\phi_{i',k',q'}\phi_{i,k,q+(\ell-n)N}}[0].\nonumber
\end{align}
Without loss of generality, let $D/\Delta\tau=N\in\mathbb{Z}$. With a change of variable $q''=q+(\ell-n)N$ and substituting the equivalent correlation in \eqref{R_phiphi_exp}, we have
\begin{align*}
	&\sum_{q\in\mathcal{Q}} \alpha_{i,k,q}e^{\mathrm{i}k\Delta\omega\ell D}R_{\phi_{i',k',q'}\phi_{i,k,q}}[(n-\ell)D]\\
	&\overset{\eqref{R_phiphi_exp}}{=}\sum_{q\in\mathcal{Q}} \alpha_{i,k,q} e^{\mathrm{i}k\Delta\omega n D}  R_{\phi_{i',k',q'}\phi_{i,k,q+(\ell-n)N}}[0]\\
	&=\sum_{q''\in\mathcal{Q}} \alpha_{i,k,q''+(n-\ell)N} e^{\mathrm{i}k\Delta\omega n D} R_{\phi_{i',k',q'}\phi_{i,k,q''}}[0]. 
\end{align*}
With the re-formulation, \eqref{sample} is re-written as below
\begin{align*}
	c_p[n] &=
	 \sum_{i'=1}^{I}\sum_{k'\in\mathcal{K}}\sum_{q'\in\mathcal{Q}} b_{p,(i',k',q')} \sum_{i=1}^{I}\sum_{k\in\mathcal{K}}\sum_{q\in\mathcal{Q}}\alpha_{i,k,q+(n-\ell)N} 
	 \\
    &\times ~ e^{\mathrm{i}k\Delta\omega n D} R_{\phi_{i',k',q'}\phi_{i,k,q}}[0]+ v_p[n].
\end{align*}
By letting $\mathbf{M}\triangleq  \mathbf{M}_{\phi\phi}[0]$ and  defining the shifted link vector $\bdsb{\alpha}[n]$ at the $n$th shift as
\begin{align}
 	\Big[\bdsb{\alpha}[n]\Big]_{(i,k,q)} \triangleq \alpha_{i,k,q+(n-\ell)N}, 
\end{align}
the observation model can be equivalently re-written as
\begin{align}
	\mathbf{c}[n] = \mathbf{B}\mathbf{M}\bdsb{\Gamma}[n]\bdsb{\alpha}[n] +\bdsb{\nu}[n].
\end{align}

\section{Proof of Proposition 1}\label{KL_max}
The pair-wise KL distance in \eqref{pair-wise_KL} can be re-written with the trace operator $\mathrm{Tr}(\cdot)$ below
\begin{align*}
	 \mathbb{D}\left(\mathcal{H}_{\mathcal{S}}\|\mathcal{H}_{\mathcal{S}'}\right)
	 &=  \frac{1}{\sigma^2}\mathrm{Tr}\left[\mathbf{M}^H\mathbf{B}^H\left(\mathbf{B}\mathbf{M}\mathbf{B}^H\right)^{-1}\mathbf{B}\mathbf{M}\mathbf{R}_{\bdsb{\beta}_{\mathcal{S}},\bdsb{\beta}_{\mathcal{S}'}}\right],
\end{align*}
where $\mathbf{R}_{\bdsb{\beta}_{\mathcal{S}},\bdsb{\beta}_{\mathcal{S}'}} \triangleq \left(\bdsb{\beta}_{\mathcal{S}}-\bdsb{\beta}_{\mathcal{S}'}\right)\left(\bdsb{\beta}_{\mathcal{S}}-\bdsb{\beta}_{\mathcal{S}'}\right)^H$. Then the average pair-wise KL distance $\overline{\mathbb{D}}$ in \eqref{avg_KL} becomes
	\begin{align}
		 \overline{\mathbb{D}}
		 &=\frac{1}{\sigma^2}\mathrm{Tr}\left[\mathbf{M}^H\mathbf{B}^H\left(\mathbf{B}\mathbf{M}\mathbf{B}^H\right)^{-1}\mathbf{B}\mathbf{M}\mathbf{R}\right],\nonumber
	\end{align}
	where $\mathbf{R} \triangleq \sum_{\mathcal{S}}\sum_{\mathcal{S}'}\gamma_{\mathcal{S},\mathcal{S}'} \mathbf{R}_{\mathcal{S},\mathcal{S}'}$ and $\mathbf{R}_{\mathcal{S},\mathcal{S}'}$ is the averaged covariance matrix of $\bdsb{\beta}_{\mathcal{S}}$ over the amplitudes
	\begin{align*}
		\mathbf{R}_{\mathcal{S},\mathcal{S}'} &=\iint P(\bdsb{\beta}_{\mathcal{S}})P(\bdsb{\beta}_{\mathcal{S}'})
		 \mathbf{R}_{\bdsb{\beta}_{\mathcal{S}},\bdsb{\beta}_{\mathcal{S}'}}\mathrm{d}\bdsb{\beta}_{\mathcal{S}}\mathrm{d}\bdsb{\beta}_{\mathcal{S}'}.
	\end{align*}
Given $P(\bdsb{\beta}_{\mathcal{S}})=\prod_{(i,k,q)\in\mathcal{S}}P(\beta_{i,k,q})$ with $\int \bdsb{\beta}_{\mathcal{S}} P(\bdsb{\beta}_{\mathcal{S}})\mathrm{d}\bdsb{\beta}_{\mathcal{S}}=\mathbf{0}$ and $\int |\beta_{i,k,q}|^2 P(\beta_{i,k,q})\mathrm{d}\beta_{i,k,q}= \sigma_\beta^2$, the averaged matrix $\mathbf{R}_{\mathcal{S},\mathcal{S}'}$ is diagonal. Furthermore, if the set of weights $\gamma_{\mathcal{S},\mathcal{S}'}$ are constant for all $\mathcal{S},\mathcal{S}'$ and the individual weighting function $P(\beta_{i,k,q})$ is identical for all $i,k,q$, it also satisfies $\mathbf{R}=\sigma_\beta^2 \mathbf{I}$ because the summation over $\mathcal{S},\mathcal{S}'$ is symmetric, and hence produces equal sum. Thus the result follows.

\section{Proof of Theorem \ref{KL_opt}}\label{theorem_3}
By analogy with Lemma \ref{ratio_trace}, we have $\mathbf{S}=\mathbf{M}$ and $\mathbf{G}=\mathbf{M}\mathbf{M}^H$ in \eqref{avg_D}. Let
\begin{align}	
	\mathbf{B} \triangleq
	\begin{bmatrix}
		\mathbf{b}_1 & \cdots &
		\mathbf{b}_P
	\end{bmatrix}^H,
\end{align}	
where $\mathbf{b}_p$ is a length-$I|\mathcal{K}||\mathcal{Q}|$ column vector with $\mathbf{b}_p=\mathbf{w}_p$. In this setting, according to Lemma \ref{ratio_trace}, the optimal $\mathbf{b}_p$ is chosen as the generalized eigenvector of the matrix pair $(\mathbf{S},\mathbf{G})$ such that $\mathbf{M}\mathbf{b}_p = \lambda_p\mathbf{M}\mathbf{M}^H\mathbf{b}_p$. Using the eigen-decomposition of $\mathbf{M}=\mathbf{U}\bdsb{\Sigma}\mathbf{U}^H$ and the property $\mathbf{U}^H\mathbf{U}=\mathbf{I}$, we have
\begin{align}
	\bdsb{\Sigma}\mathbf{U}^H\mathbf{b}_p =
	 \lambda_p\bdsb{\Sigma}\bdsb{\Sigma}^H\mathbf{U}^H\mathbf{b}_p,\quad p=1,\cdots,P.
\end{align}
If we choose $\mathbf{b}_p=\mathbf{u}_p$, where $\mathbf{u}_p$ is the $p$th column in the matrix $\mathbf{U}$, then the above relationship holds for all $p=1,\cdots,P$ as long as $P\leq \mathrm{rank}(\bdsb{\Sigma})$ because $\mathbf{u}_{i}^H\mathbf{u}_{j}=\delta[i-j]$. This gives
\begin{align*}
	&\mathrm{L.H.S.~:}~~\sigma_p\mathbf{U}^H\mathbf{u}_p = \sigma_p\mathbf{e}_p, \\
	 &\mathrm{R.H.S.~:}~~\lambda_p\bdsb{\Sigma}\bdsb{\Sigma}^H\mathbf{U}^H\mathbf{u}_p=\lambda_p\sigma_p^2\mathbf{e}_p,
\end{align*}
leading to a generalized eigenvalue of $\lambda_p=1/\sigma_p$, where $\sigma_p > 0$ is the $p$th eigenvalue in $\bdsb{\Sigma}$ and $\mathbf{e}_p$ is the canonical basis with $1$ in the $p$th entry and $0$ otherwise. Denote by $\bdsb{\Sigma}_P$ and $\mathbf{U}_P$ the principal eigenvalue and eigenvector matrix. Then the optimal $\mathbf{B}$ is chosen as 
\begin{align}
	\mathbf{B}=\bdsb{\Xi}_P\mathbf{U}_P^H, 
\end{align}	
where $\bdsb{\Xi}_P$ is an arbitrary non-singular $P\times P$ matrix. According to \eqref{avg_D}, this choice gives
\begin{align*}
	 \overline{\mathbb{D}}
	 =& \frac{\sigma_\beta^2}{\sigma^2}\mathrm{Tr}\left(\bdsb{\Sigma}_P^H\underbrace{\bdsb{\Xi}_P^H \bdsb{\Xi}_P^{-H}}_{=\mathbf{I}}\bdsb{\Sigma}_P^{-1}\underbrace{\bdsb{\Xi}_P^{-1}\bdsb{\Xi}_P}_{=\mathbf{I}}\bdsb{\Sigma}_P\right) = \frac{\sigma_\beta^2}{\sigma^2}\sum_{p=1}^{P}\sigma_p,	
\end{align*}
which is independent of $\bdsb{\Xi}_P$. If the principal eigenvectors $\mathbf{U}_P$ are unique, the above $\mathbf{B}$ maximizes the average KL distance $\overline{\mathbb{D}}$. This choice of $\mathbf{B}$ in general spreads out the individual KL distance, while the occurence of the events $\mathbb{D}\left(\mathcal{H}_{\mathcal{S}}\|\mathcal{H}_{\mathcal{S}'}\right) = 0$ is analyzed below. So is the case when $\mathbf{U}_P$ is not unique.

Now we examine the occurrence of $\mathbb{D}\left(\mathcal{H}_{\mathcal{S}}\|\mathcal{H}_{\mathcal{S}'}\right) = 0$. Let $\bdsb{\beta}_{\mathcal{S}\cup\mathcal{S}'}=\left(\bdsb{\beta}_{\mathcal{S}}-\bdsb{\beta}_{\mathcal{S}'}\right)$ be a sparse vector with $|\mathcal{S}|,|\mathcal{S}'|\leq s$, and $s\leq |\mathcal{I}|R$. Substituting $\mathbf{B} = \bdsb{\Xi}_P\mathbf{U}_P^H$ back to \eqref{pair-wise_KL} and simplifying the expression, the individual KL distance is
\begin{align}
	 \mathbb{D}\left(\mathcal{H}_{\mathcal{S}}\|\mathcal{H}_{\mathcal{S}'}\right)
	&= \frac{1}{\sigma^2}\bdsb{\beta}_{\mathcal{S}\cup\mathcal{S}'}^H\mathbf{U}_P\bdsb{\Sigma}_P\mathbf{U}_P^H\bdsb{\beta}_{\mathcal{S}\cup\mathcal{S}'},\\ &~~~~\forall \mathcal{S}\neq\mathcal{S}',~|\mathcal{S}|,|\mathcal{S}'|\leq s.
\end{align}
Since $\bdsb{\beta}_{\mathcal{S}\cup\mathcal{S}'}$ is at most a $2s$-sparse vector, thus $\mathbb{D}\left(\mathcal{H}_{\mathcal{S}}\|\mathcal{H}_{\mathcal{S}'}\right)$ is bounded away from zero as long as any $2s$-sparse vectors do not fall into the null space of the matrix $\mathbf{U}_P^H$, which implicitly implies $P\geq 2s$. In order to minimize the occurrence of the event $\mathbb{D}\left(\mathcal{H}_{\mathcal{S}}\|\mathcal{H}_{\mathcal{S}'}\right) = 0$ given a certain level of signal sparsity $s$, it is equivalent to maximizing the kruskal rank of the matrix $\mathbf{U}_P^H$ such that the matrix $\mathbf{B}$ can recover any $s$-sparse vector $\bdsb{\beta}_{\mathcal{S}}$ with $s$ as large as possible. This is consistent with the popular results in compressed sensing, therefore when the solution obtained from the optimization is not unique, one can use this as a criterion to choose the best candidate from the solutions of $\mathbf{B}$ that maximize the average KL distance.

\bibliographystyle{IEEEtran}




\bibliography{../../ref_general,../../ref_CS_application,../../ref_CS_theory_alg,../../ref_FRI_UoS,../../ref_SS_comm}

\end{document}